\documentclass[11pt,letter]{article}

\usepackage[utf8]{inputenc}
\usepackage{amsmath} 
\usepackage{amssymb} 
\usepackage{amsfonts} 
\usepackage{amscd} 
\usepackage{latexsym}
\usepackage{wasysym}
\usepackage{graphicx}
\usepackage{xcolor}
\usepackage{hyperref}
\usepackage[all,cmtip]{xy}
\usepackage{bbold}

\usepackage{amssymb}
\usepackage{amsthm}
\usepackage{enumitem}
\usepackage{mathtools,bbm}
\usepackage{xpatch}
\usepackage{version,xspace}

\usepackage[margin=1in]{geometry}
\usepackage{hyperref}%
\hypersetup{colorlinks=true, linkcolor=blue, breaklinks=true, urlcolor=blue}
\usepackage{url,doi}

\usepackage[utf8]{inputenc}
\usepackage{enumitem}
\usepackage{mathtools,bbm}
\usepackage{mathtools,bbm}
\usepackage{xpatch}

\usepackage{lipsum}
\usepackage{mathtools}
\usepackage{cuted}
\usepackage{mathtools,bbm}

\usepackage[caption=false,font=footnotesize]{subfig}
\usepackage[font=small]{caption}
\graphicspath{{./img/}{./fig}}

\usepackage[export]{adjustbox} 

\usepackage{enumitem}
\setlist{nosep}

\newcommand{\R}{\mathbb{R}}
\newcommand{\C}{\mathcal{C}}

\newcommand{\mU}{\mathcal{U}}
\newcommand{\mL}{\mathcal{L}}

\newcommand{\PP}{\mathcal{P}^2}

\newcommand{\PPd}{\mathcal{P}^2(\R^d)}

\newcommand{\gamopt}{\gamma^{\operatorname{opt}}}

\newcommand{\Pa}{\mathcal{P}}
\newcommand\norm[1]{\left\lVert#1\right\rVert}
\newcommand\bignorm[1]{\big\lVert#1\big\rVert}
\newcommand{\supp}{\operatorname{supp}}
\newcommand{\bx}{\mathbf{x}}
\newcommand{\by}{\mathbf{y}}

\newtheorem{conjecture}{Conjecture}

\newcommand{\until}[1]{\{1,\dots, #1\}}

\newtheorem{theorem}{Theorem}[section]
\newtheorem{corollary}[theorem]{Corollary}

\newtheorem{proposition}[theorem]{Proposition}
\newtheorem{remark}[theorem]{Remark} 

\newtheorem{definition}{Definition}[section]

\newcommand{\vect}[1]{\mathbbold{#1}}

\newcommand{\diag}{\operatorname{diag}}

\newcommand{\setdef}[2]{\{#1 \; | \; #2\}}

\usepackage{algpseudocode,algorithm,algorithmicx}

\usepackage{hyperref}%
\hypersetup{colorlinks=true, linkcolor=blue, breaklinks=true, urlcolor=blue}

\newcommand\oprocendsymbol{\hbox{$\square$}}
\newcommand\oprocend{\relax\ifmmode\else\unskip\hfill\fi\oprocendsymbol}

\DeclareSymbolFont{bbold}{U}{bbold}{m}{n}
\DeclareSymbolFontAlphabet{\mathbbold}{bbold}

\newcounter{saveenum}

\newcommand\blfootnote[1]{%
  \begingroup
  \renewcommand\thefootnote{}\footnote{#1}%
  \addtocounter{footnote}{-1}%
  \endgroup
}

\title{Distributed Wasserstein Barycenters \\ via Displacement Interpolation}
\author{Pedro Cisneros-Velarde \and Francesco Bullo}
\date{}

\begin{document}
\maketitle
\begin{abstract}
  Consider a multi-agent system whereby each agent has an initial
  probability measure.
  In this paper, we propose a distributed algorithm based upon stochastic,
  asynchronous and pairwise exchange of information and displacement
  interpolation in the Wasserstein space.
  We characterize the evolution of this algorithm and prove it computes the
  Wasserstein barycenter of the initial measures under various conditions.
  One version of the algorithm computes a standard Wasserstein barycenter,
  i.e., a barycenter based upon equal weights; and the other version
  computes a randomized Wasserstein barycenter, i.e., a barycenter based
  upon random weights for the initial measures.
  Finally, we specialize our algorithm to Gaussian distributions and draw a
  connection with opinion dynamics.
\end{abstract}

\blfootnote{
  This work is supported by the U. S. Army Research Laboratory and the
  U. S. Army Research Office under grant number W911NF-15-1-0577. The
    views and conclusions contained in this document are those of the
    authors and should not be interpreted as representing the official
    policies, either expressed or implied, of the Army Research Laboratory
    or the U.S. Government.}
\blfootnote{Pedro Cisneros-Velarde (e-mail:
    pacisne@gmail.com) and Francesco Bullo (e-mail:
    bullo@ucsb.edu) are with the University of California, Santa Barbara.}

\section{Introduction}

\paragraph*{Problem statement and motivation} 
There has been strong interest in the theoretical study and practical
application of Wasserstein barycenters over the last decade.  In this
paper, we characterize the evolution of a distributed system where all the
computing units or \emph{agents} hold a probability measure, interact
through pairwise communication by performing \emph{displacement
  interpolations} in the Wasserstein space. These pairwise interactions are
asynchronous 
and stochastic.  We study the
conditions under which the agents' measures will asymptotically achieve
consensus and, additionally, consensus on a Wasserstein barycenter of the
agents' initial measures. 
We are interested in computing both the
standard Wasserstein barycenter and randomized weighted versions of it --
as a result of the stochastic interactions.  
We consider both
undirected and directed communication graphs.  To the best of our
knowledge, these problems have not been studied in the literature on the
distributed computation of Wasserstein barycenters.

Asynchronous pairwise algorithms are inherently robust to communication
failures and do not require synchronization of the whole multi-agent
system. Pairwise interactions 
may potentially reduce 
the local computation complexity of each agent.
Indeed,
displacement interpolations have the practical advantage that they may have
a closed form expression, e.g., in the Gaussian case.

\paragraph*{Wasserstein barycenters and their applications}
The Wasserstein barycenter of a set of measures can be interpreted as an
interpolation or weighted Fr\'{e}chet mean of multiple measures in the
Wasserstein space~\cite{VMP-YZ:20}. In this interpolation, each measure
has an associated positive \emph{weight} that indicates its importance in
the computation of the barycenter. 
When all weights are equal, we obtain the
\emph{standard} Wasserstein barycenter; otherwise, we obtain a
\emph{weighted} one.

There has been a strong interest in the theoretical study of Wasserstein
barycenters over the last decade; e.g., uniqueness results and connections to multi-marginal 
optimal transport~\cite{MA-GC:11};
interpolation of discrete measures with finite support and 
connections 
with linear programming~\cite{GC-AO-EO:15,EA-SB-JM:16}; the
characterization of the barycenter as a fixed point of an operator and 
its computation~\cite{PCAE-EdB-JACA-CM:16};
the study of consistency and other statistical properties~\cite{TLG-JML:17}.  For further information, we refer to 
~\cite{VMP-YZ:20}.

Along with the theoretical progress, many applications of Wasserstein
barycenters have emerged, as well as 
numerical approaches
for computing them.
Wasserstein barycenters
have found applications in economics~\cite{Gc-IE:10}, image processing~\cite{JR-GP-JD-MB:12,YM:20}, computer graphics~\cite{NB-GP-MC:16}, physics~\cite{GB-LDP-PGG:12}, statistics~\cite{SS-VC-QD-DD:15,VS-MC:15}, machine learning~\cite{NC-RF-DT-AR:17,MAS-MH-NB-FN-DC-MC-GP-JLS:18}, signal processing~\cite{ANB:14,MB-PKW-UDH:15}, and biology~\cite{SG-JML-EM:13}. 
Examples of computational approaches include: exact algorithms~\cite{GC-AO-EO:15,SC-EC-JS:18}, algorithms that use 
entropic regularization~\cite{MC-AD:14,MC-GP:16}, and algorithms based on approximations of Wasserstein distances~\cite{NB-JR-GP-HP:15}. 
Finally, the particular case of interpolating two measures, i.e., the
\emph{displacement interpolation,} 
has 
applications in partial differential equations and
geometry~\cite{CV:09,FS:15}, and fluid mechanics~\cite{JDB-YB:00}.

Moreover, the Wasserstein barycenter has been interpreted as a denoised version of an original signal whose sensor measurements are each of the noisy probability distributions that are being interpolated; thus, it 
has found 
applications as an \emph{information fusion} algorithm~\cite{SG-JML-EM:13,ANB:14,NB-GP-MC:16,SC-EC-JS:18}. 
%
%
In this setting, a randomized Wasserstein barycenter, i.e., one which randomly weight each sensor, could be used to provide different estimates of the true signal when interpolating measurements of sensors of unknown accuracy or noise level.

Finally, we mention that the problem studied in this paper directly contributes to
the literature on consensus, a research area which has attracted great interest from the systems and
controls community. Specifically, we contribute to 
the fields of randomized consensus algorithms -- e.g., see~\cite[Chapter~13]{FB:22} and references therein -- and of consensus in spaces other than the
Euclidean space -- e.g., see~\cite{RS:11b,IM-JSB:15,ANB-AD:21}.  
Moreover, our work contributes to the field of opinion dynamics,
where classic opinion models also use stochastic asynchronous pairwise interactions~\cite{GD-DN-FA-GW:00,DA-AO:11}. 
Indeed, in 
our paper, 
we argue that 
displacement interpolation is a more suitable modeling approach for the
non-Bayesian updating of individual's beliefs 
in a social network,
than
classic averaging approaches in the literature.

\paragraph*{Distributed algorithms for Wasserstein barycenters}
To the best of our knowledge, there is only a recent and growing literature
on distributed algorithms for Wasserstein barycenters.  The idea of
computing Wasserstein barycenters in a distributed way was first pioneered
by Bishop and Doucet in their work~\cite{ANB-AD:14} and its very recent extension~\cite{ANB-AD:21}. Their work formally shows consensus
towards the Wasserstein barycenter of the agents' initial
measures. In order to compute such consensus, each agent needs to
fully compute the Wasserstein barycenter resulting from its own measure and
the measures from all its neighbors at each iteration according to some time-varying graph. 
The work~\cite{ANB-AD:21} focuses on the case of probability measures on
the real line, with the 
communication between agents being deterministic, but flexible enough to
consider both synchronous and asynchronous deterministic updating.  It
 assumes agents are connected by an undirected graph.

The recent work~\cite{CAU-DD-PD-AG-AN:18} focuses on the design and
distributed implementation of a numerical solver that approximates the
standard Wasserstein barycenter when all the measures are discrete, through
the use of entropic regularization. Moreover, the recent
work~\cite{PD-DD-AG-CU-AN:18} from the same authors proposes another
distributed solver for an approximate Wasserstein barycenter with the
difference that the agents' measures may correspond to continuous
distributions. Indeed, its framework is \emph{semi-discrete}, in that the
measures to be interpolated can be continuous, but the sought measure that
serves as a proxy for the barycenter is restricted to be a discrete
measure with finite support. Therefore, we observe that the distributed
algorithms from both works~\cite{CAU-DD-PD-AG-AN:18,PD-DD-AG-CU-AN:18}
compute an approximate or a proxy of the true barycenter.  
Both 
works require synchronous updating
and all the 
computations are performed over an undirected graph.

Our paper is more in line with the spirit of~\cite{ANB-AD:14,ANB-AD:21}, in the sense
that we propose a theoretical formulation and analysis that prove how to
generate Wasserstein barycenters from distributed computations. We do not
propose specific designs of numerical solvers for the local computations of
the agents, as it is instead performed
in~\cite{CAU-DD-PD-AG-AN:18,PD-DD-AG-CU-AN:18}. Indeed, since the local
computations in our algorithm are displacement interpolations at every time
step, any numerical method that can solve optimal transport problems can be
used, including for example any of the numerical algorithms mentioned
above.

\paragraph*{Contributions}
In this paper we propose the algorithm \emph{PaWBar}
(\emph{Pa}irwise distributed algorithm for \emph{W}asserstein
\emph{Bar}ycenters), where the agents update their
measures via pairwise stochastic and asynchronous interactions implementing
displacement interpolations. The algorithm has 
a
\emph{directed} and a \emph{symmetric} version. As main contribution of this paper, 
we establish conditions
under which both versions compute randomized and standard Wasserstein
barycenters respectively.
In the directed case, we prove that every time the algorithm is run, a
barycenter with random convex weights is asymptotically generated as a
result of the stochastic selection of the pairwise interactions. 
It is easy to characterize the first two moments of these random weights.
On the other hand, in the symmetric case, 
although 
the
interactions are stochastic, we prove that the asymptotically computed
Wasserstein barycenter is the standard one (with probability one).
In contrast to the works~\cite{CAU-DD-PD-AG-AN:18,PD-DD-AG-CU-AN:18}, our
algorithm does not require all the agents to synchronously update their
measures at every time step. 
Moreover, 
our framework provides
convergence guarantees towards the computation of the barycenter
independently from the numerical implementation of the local computations.
We also remark that work~\cite{ANB-AD:21} is different from ours because:
(i) it only focuses on measures on the real line $\R$, while we
  consider $\R^d$, $d\geq 1$; (ii) its underlying communication graph is
  undirected at all time-steps and its changes are deterministic; (iii) it
  considers local computations of the full Wasserstein barycenter between
  agent and its neighbors; and (iv) we present sufficient conditions for
  the computation of the standard Wasserstein barycenter.
%


We now elaborate on the convergence results. We first prove convergence to a randomized or standard Wasserstein
barycenter for a class of discrete measures on $\R^d$, $d\geq 1$. In
particular, we show that the obtained barycenter interpolates the agents'
measures attained at some random finite time. However, if the initial
measures are sufficiently close in the Wasserstein space, then such time is
zero with probability one, i.e., there is an interpolation of the initial
measures. For the 
case where these discrete measures are on
$\R$, the interpolation of the initial measures occurs with probability one 
irrespective of how distant they initially are from each other.

We then prove convergence to a randomized or standard Wasserstein
barycenter for a class of measures that are absolutely continuous with
respect to the Lebesgue measure on $\R^d$. As corollaries, we prove
convergence of continuous probability distributions on the real line, and
of a class of multivariate Gaussian distributions. 
In the Gaussian case, 
we also provide simpler closed form expressions for
the computations of the PaWBar algorithm, and a simplified expression of
the converged barycenter. We also conjecture that the convergence to
Wasserstein barycenters holds for general absolutely continuous measures,
and present supporting numerical evidence for the general multivariate Gaussian
case.

Moreover, in all the cases mentioned above, the convergence results are
proved under 
general communication graphs: 
a strongly connected digraph
for the directed PaWBar
algorithm and a connected undirected graph for the symmetric algorithm.  For randomized
barycenters, we characterize their random convex coefficients by the limit
product of random stochastic matrices.

Finally, we prove a general consensus result for the case 
of arbitrary
initial measures on $\R^d$ 
by making 
strong use
of general 
geodesic properties of the Wasserstein space. The results are
proved over a cycle graph for the directed PaWBar algorithm and 
a line 
graph for the symmetric case. 
On $\R$, using a result from~\cite{ANB-AD:21}, our algorithms achieve 
consensus under 
general graphs.

\paragraph*{Paper organization}
Section~\ref{sec:prelim} has notation and preliminary
concepts. Section~\ref{sec:prop-alg} has the proposed~\emph{PaWBar
  algorithm} and its theoretical
analysis. Section~\ref{sec:proof} presents the proofs for
Section~\ref{sec:prop-alg}. Section~\ref{sec:app-ex} presents the
connection between our algorithm and opinion dynamics. 
Section~\ref{sec:concl} is the conclusion.

\section{Notation and preliminary concepts}
\label{sec:prelim} 

Let $z=(z_1,\dots,z_n)^\top\in\R^n$ denote a vector. 
Let $\norm{\cdot}_2$ denote the Euclidean
distance. 
The standard unit vector $e_i\in\{0,1\}^n$ has one in its $i$th
entry.
Let $\vect{1}_n,\vect{0}_n\in\R^n$ be the all-ones and
all-zeros vectors respectively, and $I_n$ be the $n\times n$ identity
matrix. Nonnegative matrix $A\in\R^{n\times n}$ is \emph{row-stochastic}
if $A\vect{1}_n=\vect{1}_n$, and \emph{doubly-stochastic} if additionally
$A^\top\vect{1}_n=\vect{1}_n$.  The operator $\circ$ the composition of
functions, and $\otimes$ the Kronecker product.

The numbers $\lambda_1,\dots,\lambda_n$ are called \emph{convex
  coefficients} if $\lambda_i\geq 0$, $i\in\until{n}$, and
$\sum_{i=1}^n\lambda_i=1$. The vector
$\lambda:=(\lambda_1,\dots,\lambda_n)^\top$ is called a \emph{convex
  vector}.

The set of agents is $V=\until{n}$, $n\geq 2$. 
The agents are connected according to the graph $G=(V,E)$; with set of nodes $V$ and set of edges $E$. When 
$E$ 
only has  
ordered pairs, i.e., $(i,j)\in E$ with 
$i,j\in V$, $G$ is a directed
graph or \emph{digraph}. Thus, $(i,j)\in E$ 
is 
a directed edge going
from $i$ to 
$j$. When 
$E$ 
only has 
unordered pairs, i.e., $\{i,j\}\in E$ 
with $i,j\in V$, $G$ is an undirected graph, and its edges 
have no sense of direction.
$G$ is \emph{weighted} when a scalar value is assigned to every edge. 
An undirected graph $G$ is a line graph  
when its nodes can be labeled as
$E=\{\{1,2\},\dots,\{n-1,n\}\}$. A digraph $G$ is a cycle 
when its nodes can be labeled as
$E=\{(1,2),\dots,(n-1,n),(n,1)\}$.  
Given any $i,j\in V$, a digraph is strongly connected
when it is possible to go from $i$ to $j$ by
traversing the edges according to their direction (e.g., a cycle graph); and 
an undirected graph is connected when it is
possible to go from $i$ to $j$ by traversing the edges in any
direction (e.g., a line graph).

We denote the set of all probability measures on $\Omega\subseteq\R^d$ by
$\mathcal{P}(\Omega)$,
and define 
$\mathcal{P}^2(\Omega)=\setdef{\mu\in\mathcal{P}(\Omega)}{\int_{\Omega}\norm{x}_2^2d\mu(x)<\infty}$.
Consider $\mu\in\mathcal{P}(\Omega)$. For $\Omega=\R$, let $F_\mu$ be the
cumulative distribution function, i.e., $F_{\mu}(x)=\mu((-\infty,x])$.
Let $\#$ be the \emph{push-forward operator}, which, for any Borel measurable map $\mathcal{M}:\Omega\to\Omega$, defines the linear operator $\mathcal{M}_{\#}:\mathcal{P}(\Omega)\to\mathcal{P}(\Omega)$  
by $(\mathcal{M}_\#\mu)(B)=\mu(\mathcal{M}^{-1}(B))$ for any Borel set $B\subseteq\Omega$. We denote the support of $\mu$ by $\supp(\mu)$.

Given $\mu,\nu\in\mathcal{P}^2(\Omega)$, the $2$-Wasserstein distance 
between $\mu$ and $\nu$ is 
\begin{equation}
\label{eq:Wass-dist}
W_2(\mu,\nu)=\left(\inf_{\gamma\in\Pi(\mu,\nu)}\int_{\Omega\times\Omega}\norm{x-y}_2^2d\gamma(x,y)\right)^{1/2}
\end{equation}
where $\Pi(\mu,\nu)$ is the set of probability measures on $\Omega\times\Omega$ with marginals $\mu$ and $\nu$, i.e., if $\gamma\in\Pi(\mu,\nu)$, then $(\pi_1)_\#\gamma=\mu$ and $(\pi_2)_\#\gamma=\nu$ with $\pi_1(x,y)=x$ and $\pi_2(x,y)=y$. The optimization problem that defines the Wasserstein distance 
is an \emph{
optimal transport problem},
and any of its solutions 
is an \emph{optimal transport plan}. Let $\gamopt(\mu,\nu)$ denote an optimal transport plan between measures $\mu$ and $\nu$. 
Given $\gamopt(\mu,\nu)$ such that 
$\nu=\mathcal{T}_{\#}\mu$, we say that $\gamopt$ solves the \emph{Monge optimal transport problem} and the map $\mathcal{T}$ is called the \emph{optimal transport map} from $\mu$ to $\nu$. The Wasserstein space of order $2$ is the space $\mathcal{P}^2(\Omega)$ endowed with the distance $W_2$. In this paper, 
we will consider $\Omega=\R^d$, $d\geq 1$.

Given convex coefficients $\lambda_1,\dots,\lambda_n$ -- also called \emph{weights} -- and 
measures to interpolate $\mu_1,\dots\mu_n\in\mathcal{P}^2(\Omega)$, $\Omega$ convex, $n\geq 2$, a \emph{Wasserstein barycenter} is defined by any solution to the convex problem
\begin{equation*}
\min_{\nu\in\mathcal{P}^2(\Omega)}\sum_{i=1}^n\lambda_iW^2_2(\nu,\mu_i).
\end{equation*}
%
The \emph{displacement interpolation} between measures $\mu,\nu\in\mathcal{P}^2(\Omega)$ is the curve $\mu_{\lambda}=(\pi_\lambda)_{\#}\gamopt(\mu,\nu)$, $\lambda\in[0,1]$, where $\pi_\lambda:\Omega\times\Omega\to\Omega$ is defined by $\pi_\lambda(x,y)=(1-\lambda)x+\lambda y$. The curve $\pi_\lambda$ is known to be a \emph{constant-speed geodesic curve} in the Wasserstein space connecting $\mu_0=\mu$ to $\mu_1=\nu$~\cite{FS:15}. Moreover, for a fixed $\lambda\in[0,1]$, it is known to be the solution to the Wasserstein barycenter problem $\min_{\rho\in\mathcal{P}_2(\Omega)}((1-\lambda)W_2^2(\rho,\mu_1)+\lambda W_2^2(\rho,\mu_2))$. When there exists an optimal transport map $\mathcal{T}$, then 
$(\pi_\lambda)_{\#}\mu=((1-\lambda)I\!d+\lambda \mathcal{T})_\#\mu$, $\lambda\in[0,1]$, where $I\!d$ is the identity operator.
 
\section{Proposed algorithm and analysis}
\label{sec:prop-alg}

\subsection{The PaWBar algorithm}

Let $\mu_i(t)\in\PPd$, $i\in V$, represent the measure of agent $i$ at time
$t\in\{0,1,\dots\}$. Our proposed \emph{PaWBar} (\emph{P}airwise distributed
\emph{a}lgorithm for \emph{W}asserstein \emph{Bar}ycenters) algorithm has
two versions.

\begin{definition}[Directed PaWBar algorithm]
\label{def:alg}
  Let $G$ be a weighted directed graph with weight $a_{ij}\in(0,1)$ for
  $(i,j)\in E$.  Assume $\mu_i(0):=\mu_{i,0}\in\PPd$ for every $i\in V$. 
  At each time $t$, execute:
  \begin{enumerate}
  \item select a random edge $(i,j)\in E$ of $G$, independently
  according to some time-invariant probability distribution, with all edges
  having a positive selection probability;
  \item update the measure of agent $i$ by
  \begin{equation}\label{f1}
    \mu_i(t+1):=
	(\pi_{a_{ij}})_\#\gamopt(\mu_i(t),\mu_j(t))
  \end{equation}
  where $\pi_{a_{ij}}:\R^d\times{\R^d}\to\R^d$ is
  defined by $\pi_{a_{ij}}(x,y)=(1-a_{ij})x+a_{ij}y$.
    \end{enumerate}
\end{definition}

\begin{definition}[Symmetric PaWBar algorithm]
\label{def:alg_s}
  Let $G$ be an undirected graph. 
  Assume $\mu_i(0):=\mu_{i,0}\in\PPd$ for every $i\in V$. 
  At each time $t$, execute:
  \begin{enumerate}
  \item select a random edge $\{i,j\}\in E$ of $G$, independently
  according to some time-invariant probability distribution, with all edges
  having a positive selection probability;
  \item update the measures of agents $i$ and $j$ by
\begin{equation}\label{f2}
  \begin{aligned}
    \mu_i(t+1)=\mu_j(t+1):=(\pi_{1/2})_\#\gamopt(\mu_i(t),\mu_j(t))
  \end{aligned}
  \end{equation}
  where $\pi_{1/2}:\R^d\times{\R^d}\to\R^d$ is 
  defined by $\pi_{1/2}(x,y)=\frac{1}{2}(x+y)$.
    \end{enumerate}
\end{definition}

\begin{remark}[Well-posedness]
The PaWBar algorithm is well-posed since 
the displacement interpolation 
provides measures in $\PPd$~\cite[Theorem~5.27]{FS:15}.
\end{remark}

\begin{remark}[Symmetry in the interpolated measure]
  Since $\pi_{1/2}(x,y)=\pi_{1/2}(y,x)$ for any $x,y\in\R^d$, the update
  rule~\eqref{f2} of the symmetric PaWBar algorithm is equivalent to
  $\mu_i(t+1):=(\pi_{1/2})_\#\gamopt(\mu_i(t),\mu_j(t))$ and
  $\mu_j(t+1):=(\pi_{1/2})_\#\gamopt(\mu_j(t),\mu_i(t))$.
\end{remark}

For simplicity, we call \emph{edge selection process} the underlying stochastic process of edge selection by the PaWBar algorithm, whose realizations are the infinite sequence of selected edges chosen every time 
the PaWBar algorithm is run. When a result is stated \emph{with probability one}, it is to be understood with respect to the induced measure by 
the edge selection process. The following concept and proposition are useful for our results.

\begin{definition}[Evolution random matrix]
\label{def:evol-rand-matr}
Consider the edge selection process from the PaWBar algorithm. Define the \emph{evolution random matrix} $A(t)$ by: 
\[
  A(t)=
  \begin{cases}
    I_n-a_{ij}e_ie_i^\top-(1-a_{ij})e_ie_j^\top,\\ \qquad \text{if } (i,j)\in E\text{ is chosen,} \\
    I_n-\frac{1}{2}(e_ie_i^\top+e_je_j^\top+e_ie_j^\top+e_je_i^\top),\\ \qquad \text{if } \{i,j\}\in E\text{ is chosen.}    
  \end{cases}
\]
\end{definition}

The following result is a direct application of~\cite[Theorem~13.1, Corollary~13.2]{FB:22}.

\begin{proposition}[Convergence of products of evolution random matrices]
\label{prop:conv_A}
Consider the PaWBar algorithm. For the directed case with a strongly connected digraph: $\lim_{t\to\infty}\prod_{\tau=0}^tA(\tau)=\vect{1}_n\lambda^\top$ for some random convex vector $\lambda$ with probability one. For the symmetric case with a connected undirected graph, $\lim_{t\to\infty}\prod_{\tau=0}^tA(\tau)=\frac{1}{n}\vect{1}_n\vect{1}_n^\top$ with probability one.
\end{proposition}

\subsection{Analysis of discrete measures}
\label{sub:discr}
%

\begin{theorem}[Wasserstein barycenters for discrete measures in $\PPd$]
\label{th:discr-gen}
Consider initial measures $\{\mu_{i,0}\}_{i\in V}$, such that $\mu_{i,0}=\frac{1}{N}\sum^N_{j=1}\delta_{x^i_{j}}$, with $x^i_1,\dots,x^i_N\in\R^d$ being distinct points; i.e., $\mu_{i,0}$ is a discrete uniform measure. 
\begin{enumerate}
\item\label{it-dis-gen-dir} Consider the directed PaWBar algorithm with an underlying strongly connected digraph $G$; then, with probability one, for any $i\in V$, 
\begin{equation}
\label{eq:conv-discr-1}
W_2(\mu_i(t),\mu_{\infty})\to 0 \text{ as }t\to\infty,
\end{equation}
where the discrete uniform measure $\mu_{\infty}$ is a 
barycenter 
\begin{equation*}
\mu_\infty\in\arg\min_{\nu\in\PPd}\sum_{i=1}^n\lambda_iW_2(\nu,\mu_i(T))^2
\end{equation*}
with $\lambda=(\lambda_1,\dots,\lambda_n)^\top$ being a random convex vector
satisfying $\prod^\infty_{\tau=1}A(\tau)=\vect{1}_n\lambda^\top$ with probability one,
and $T\geq 0$ being some finite random time. 
\item\label{it-dis-gen-sym}Consider the symmetric PaWBar algorithm with an underlying connected undirected graph $G$; then, with probability one, for any $i\in V$, 
\begin{equation}
\label{eq:conv-discr-2}
W_2(\mu_i(t),\mu_{\infty})\to 0\text{ as }t\to\infty,
\end{equation}
where the discrete uniform measure $\mu_{\infty}$ 
is a 
barycenter 
\begin{equation}
\label{eq:dis-gen-infty}
\mu_\infty\in\arg\min_{\nu\in\PPd}\sum_{i=1}^nW_2(\nu,\mu_i(T))^2
\end{equation}
with $T\geq 0$ being some finite random time. 
\end{enumerate}
In either case~\ref{it-dis-gen-dir} or~\ref{it-dis-gen-sym}, there exists $\epsilon>0$ such that if $\max_{i,j\in V}W_2(\mu_{i,0},\mu_{j,0})<\epsilon$, then $T=0$ with probability one.
\end{theorem}

\begin{corollary}[Wasserstein barycenters for discrete measures in $\PP(\R)$]
\label{co:discr-uni}
Consider initial measures $\{\mu_{i,0}\}_{i\in V}$, such that $\mu_{i,0}=\frac{1}{N}\sum^N_{j=1}\delta_{x^i_{j}}$, with $x^i_1,\dots,x^i_N\in\R$ such that $x^i_1<\dots< x^i_N$. 
Then, the directed, resp. symmetric, PaWBar algorithm computes a randomized, resp. standard, Wasserstein barycenter of the initial measures under a strongly connected digraph, resp. connected undirected graph. 
\end{corollary}

\begin{remark}[Discussion of our results]
\begin{enumerate}
\item The setting of Theorems~\ref{th:discr-gen} and Corollary~\ref{co:discr-uni} 
has found applications in computational geometry, computer graphics and digital image processing; e.g., see~\cite{JR-GP-JD-MB:12,MC-AD:14,NB-JR-GP-HP:15}. 
\item In Theorem~\ref{th:discr-gen}, if all initial measures are sufficiently close in the Wasserstein space, then the PaWBar algorithm computes a 
barycenter. This sufficient condition is not a problem in practical applications where the barycenter is used as an interpolation among measures that are known to be similar (e.g., measurements of the same object under noise). On the other hand, Corollary~\ref{co:discr-uni} tells us that the initial measures in $\Pa(\R)$ can be arbitrarily distant from each other.
\end{enumerate}
\end{remark}

\subsection{Analysis of absolutely continuous measures}
\label{sub:cont}

We consider measures that are absolutely continuous with respect to the
Lebesgue measure.  For any such measures $\mu,\nu\in\PPd$, there exists a
unique optimal transport map from $\mu$ to $\nu$, which we denote
by $T_\mu^\nu$; we also denote $T^\mu_\mu=I\!d$. It is also known that there exists
a unique Wasserstein barycenter when all the interpolated measures are 
of this class~\cite{MA-GC:11}.

We focus on the following class of
measures that form a \emph{compatible collection} (based on~\cite[Definition~2.3.1]{VMP-YZ:20}): 
a collection of absolutely
continuous measures $\C\subset\PPd$ 
where 
for all
$\nu,\mu,\gamma\in\mathcal{C}$, we have $T_{\gamma}^\mu\circ
T_{\nu}^\gamma=T_{\nu}^\mu$. It is known that a displacement interpolation between any two absolutely continuous measures results in a curve of absolutely continuous measures~\cite{FS:15}.
%
This motivates the following definition: we say a compatible collection is \emph{closed under interpolation} whenever the union of this set and any measure resulting from the displacement interpolation between any of its elements results in another compatible collection with the same property.

\begin{theorem}[Wasserstein barycenters for absolutely continuous measures in $\PPd$]
\label{th:cont-gen}
Consider initial measures $\{\mu_{i,0}\}_{i\in V}$ that are absolutely continuous with respect to the Lebesgue measure and that form a compatible collection 
closed under interpolation.
Let $\gamma\in\{\mu_{i,0}\}_{i\in V}$.
\begin{enumerate}
\item\label{it-cont-gen-dir} Consider the directed PaWBar algorithm with an underlying strongly connected digraph $G$; then, with probability one, for any $i\in V$, 
\begin{equation*}
W_2(\mu_i(t),\mu_{\infty})\to 0\text{ as }t\to\infty,
\end{equation*}
where the absolutely continuous measure $\mu_{\infty}=\left(\sum^n_{j=1}\lambda_j T_{\gamma}^{\mu_{j,0}}\right)_{\#}\gamma$ is the barycenter
\begin{equation}
\label{eq:dis-cont-infty}
\mu_\infty=\arg\min_{\nu\in\PPd}\sum_{i=1}^n\lambda_iW_2(\nu,\mu_{i,0})^2
\end{equation}
with $\lambda=(\lambda_1,\dots,\lambda_n)^\top$ being a random convex vector satisfying $\prod^\infty_{\tau=1}A(\tau)=\vect{1}_n\lambda^\top$ with probability one.
\item\label{it-cont-gen-sym}Consider the symmetric PaWBar algorithm with an underlying connected undirected graph $G$; then, with probability one, for any $i\in V$, 
\begin{equation*}
W_2(\mu_i(t),\mu_{\infty})\to 0\text{ as }t\to\infty,
\end{equation*}
where the absolutely continuous measure $\mu_{\infty}=\left(\frac{1}{n}\sum^n_{j=1} T_{\gamma}^{\mu_{j,0}}\right)_{\#}\gamma$ is the barycenter
\begin{equation*}
\mu_\infty=\arg\min_{\nu\in\PPd}\sum_{i=1}^nW_2(\nu,\mu_{i,0})^2.
\end{equation*}
\end{enumerate}
\end{theorem}

The following corollary considers examples of measures relevant to our previous theorem. We use the term distribution and measure interchangeably for well-known probability measures with continuous distributions.

\begin{corollary}[Wasserstein barycenters for classes of absolutely continuous measures]
\label{co:discr-gen-uni}
Consider that initially either 
\begin{enumerate}
\item\label{it:co-1} all agents have a probability measure in $\PP(\R)$ with continuous distribution; or
\item \label{it:co-2} one agent has the standard Gaussian distribution on $\PPd$ and any other agent $i\in V$ has a Gaussian distribution $\mu_{i,0}=\mathcal{N}(\vect{0}_d,\Sigma_{i,0})$, positive definite matrix $\Sigma_{i,0}\in\R^{d\times d}$; and there exists an orthogonal matrix $U$ such that $D_{k,o}=U\Sigma_{k,0}U^\top$ is diagonal for all $k$.
\end{enumerate}
Consider the directed PaWBar algorithm with an underlying strongly connected digraph $G$. 
Then, with probability one, $W_2(\mu_i(t),\mu_{\infty})\to 0$ as $t\to\infty$ for any $i\in V$,
where $\mu_{\infty}$ is the Wasserstein barycenter of the initial measures. In particular,
\begin{itemize}
\item for case~\ref{it:co-1}, $\mu_{\infty}=\left(\sum^n_{j=1}\lambda_j F^{-1}_{\mu_{j,0}}\circ F_{\mu_{i,0}}\right)_{\#}\mu_{i,0}=\left(\sum^n_{j=1}\lambda_j F^{-1}_{\mu_{j,0}}\right)_{\#}\mathcal{L}$, with $\mathcal{L}$ being the Lebesgue measure on $[0,1]$; and
\item for case~\ref{it:co-2}, $\mu_{\infty}=\mathcal{N}(\vect{0}_d,\Sigma_\infty)$ with 
$\Sigma_{\infty}\in\R^{d\times d}$ being a positive definite matrix that satisfies $\Sigma_{\infty}=\sum_{j=1}^n\lambda_j(\Sigma^{1/2}_{\infty}\Sigma_{j,0}\Sigma^{1/2}_{\infty})^{1/2}$; 
\end{itemize} 
with $\lambda=(\lambda_1,\dots,\lambda_n)^\top$ being a random convex vector such that $\prod^\infty_{\tau=1}A(\tau)=\vect{1}_n\lambda^\top$ with probability one.

Moreover, all the previous results also hold for the symmetric PaWBar algorithm when $G$ is a connected undirected graph, with the difference that the barycenters are now the standard one, i.e., with 
$\lambda=(\frac{1}{n},\dots,\frac{1}{n})^\top$ in the previous bullet points.
\end{corollary}

For Gaussian distributions, since any displacement interpolation results in Gaussian distributions with a closed-form expression~\cite{YC-TTG-AT:19}, the PaWBar algorithm can be written as follows.

\begin{definition}[PaWBar algorithm for Gaussian distributions]
\label{def:gauss}
Assume any agent $i\in V$ has an initial distribution $\mu_{i,0}=\mathcal{N}(m_{i,0},\Sigma_{i,0})$ with $m_{i,0}\in\R^d$ and $\Sigma_{i,0}\in\R^{d\times d}$ being a positive definite matrix. 
At any time $t$, let $m_i(t)$ and $\Sigma_i(t)$ be the mean and covariance matrix associated to agent $i\in V$.
\begin{enumerate}
\item\label{it:gauss-dir} For the directed PaWBar algorithm, if $(i,j)\in E$ is selected at time $t$, update
  the Gaussian distribution of agent $i$ according to:
  \begin{equation}\label{eq:g1}
  \begin{aligned}
   & m_i(t+1):=\; (1-a_{ij})m_i(t)+a_{ij}m_j(t),\\
	&\Sigma_i(t+1)=\Sigma_i(t)^{-1/2}((1-a_{ij})\Sigma_i(t)\\
	&\;+a_{ij}(\Sigma_i(t)^{1/2}\Sigma_j(t)\Sigma_i(t)^{1/2} )^{1/2})^2\Sigma_i(t)^{-1/2}.
  \end{aligned} 
  \end{equation}
\item \label{it:gauss-sym} For the symmetric PaWBar algorithm, if $\{i,j\}\in E$ is selected at time $t$, update
the Gaussian distributions of agents $i$ and $j$ by using~\eqref{eq:g1} with $a_{ij}=1/2$ and $m_i(t+1)=m_j(t+1)$, $\Sigma_i(t+1)=\Sigma_j(t+1)$.
%
\end{enumerate}
\end{definition}
%


We remark that the work~\cite{PCAE-EdB-JACA-CM:16} proposes a non-distributed iterative algorithm tailored to compute the Wasserstein barycenter of Gaussian distributions. However, to the best of our knowledge, the PaWBar algorithm is the first one proposing a distributed computation of randomized and standard Gaussian barycenters. 
%

\begin{remark}[Further characterizations 
of the randomized Wasserstein barycenter]
The results in~\cite{ATS-AJ:10} can be  
applied to characterize 
the mean and covariance matrix associated to the random convex vector present in the randomized Wasserstein barycenter in Theorem~\ref{th:discr-gen}, Corollary~\ref{co:discr-uni}, Theorem~\ref{th:cont-gen}, and Corollary~\ref{co:discr-gen-uni}, which 
numerically depends on the values of the time-invariant probabilities associated with the edge selection process.
\end{remark}

We propose the following 
conjecture.

\begin{conjecture}[Computation under more general absolutely continuous measures]
\label{conj:contdis}
The convergence results of the PaWBar algorithm in Theorem~\ref{th:cont-gen} also hold for general absolutely continuous measures. 
\end{conjecture}

We provide some numerical evidence that Conjecture~\ref{conj:contdis} is true at least for the case where 
all agents initially have multivariate Gaussian distributions that do not form a compatible collection. 
In the numerical evidence presented in Figure~\ref{fig:1-Gauss} and Figure~\ref{fig:2-Gauss}, we use the update 
rules from  
Definition~\ref{def:gauss} and 
only focus on the evolution of the agents' covariance matrices (the mean vectors evolve linearly and are easy to verify they converge to the mean of the Wasserstein barycenter). 
%
%
Although we only present results for the directed PaWBar algorithm, similar results supporting our
conjecture were obtained for the symmetric case too.
%

\begin{figure}[ht]
  \centering
  \subfloat{\label{f:1a}\includegraphics[width=0.25\linewidth,valign=c]{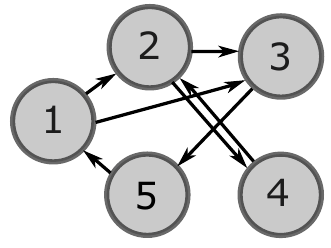}}  
    \hfil
  \subfloat{\label{f:1b}\includegraphics[width=0.7\linewidth,valign=c]{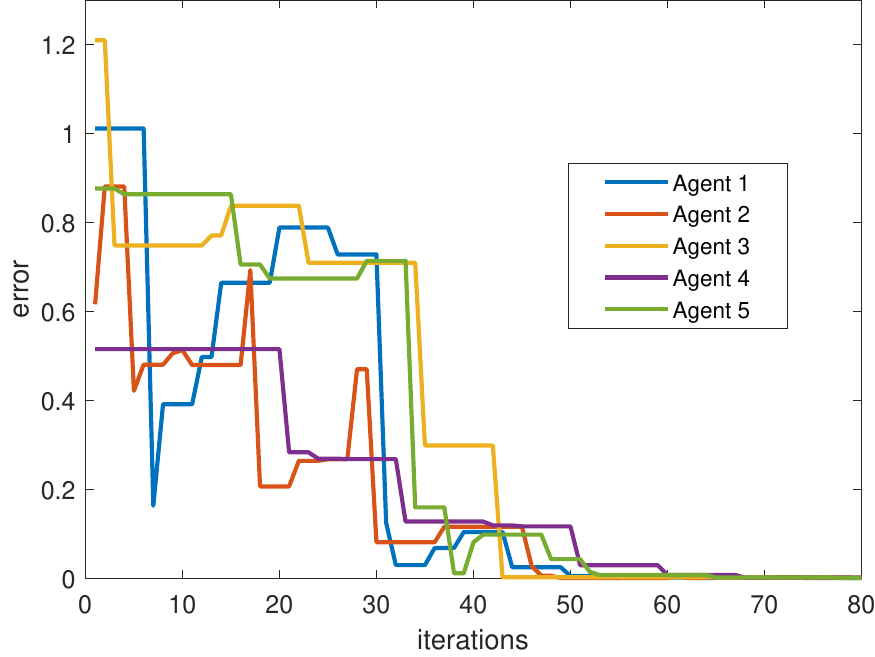}}
\caption{Consider five agents that initially have multivariate Gaussian distributions on $\R^5$, with their covariance matrices being randomly generated. On the left, we present the underlying digraph over which the PaWBar algorithm is run. The weight associated to all edges is $0.75$. We first fix a realization of the edge selection process. 
Then, 
we compute the covariance matrix $\Sigma_{\infty}$ of the Wasserstein barycenter
following the centralized numerical 
scheme proposed in~\cite{PCAE-EdB-JACA-CM:16}. On the right, each of the five plotted curves corresponds to the evolution of the error quantity $\norm{\Sigma_i(t)-\Sigma_{\infty}}_{F}$ for each agent $i\in\until{5}$, where $\Sigma_i(t)$ is the value of agent $i$'s  covariance matrix at iteration $t$, and $\norm{\cdot}_{F}$ is the Frobenius norm. All agents asymptotically reach consensus and 
converge to 
the randomized  Wasserstein barycenter, thus giving evidence for the veracity of Conjecture~\ref{conj:contdis} at least for the Gaussian case.
}\label{fig:1-Gauss}
\end{figure} 

\begin{figure}[ht]
  \centering
  \subfloat{\label{f:2a}\includegraphics[width=0.25\linewidth,valign=c]{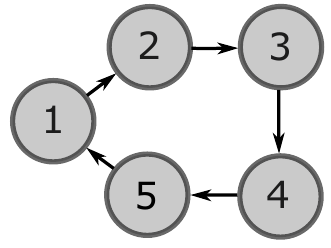}}
    \hfil
  \subfloat{\label{f:2b}\includegraphics[width=0.7\linewidth,valign=c]{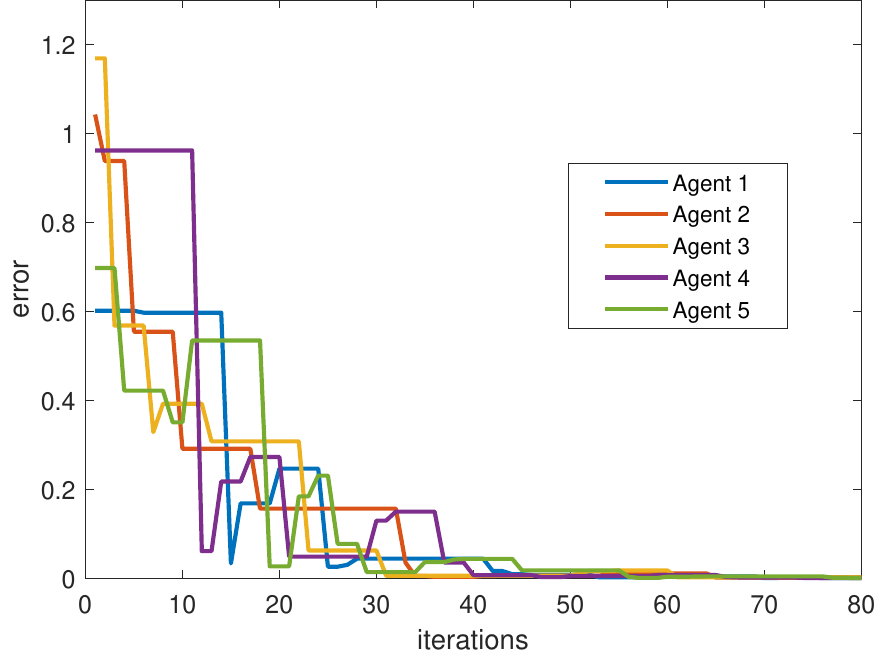}}
\caption{Consider five agents that initially have multivariate Gaussian distributions on $\R^5$, with their covariance matrices being randomly generated. The setting and methodology for computing the plot on the right is similar to the one described in Figure~\ref{fig:1-Gauss}, with the difference that now the underlying digraph is a cycle (as seen on the left). 
All agents asymptotically 
compute the 
randomized Wasserstein barycenter.
}\label{fig:2-Gauss}
\end{figure} 

\subsection{Analysis of general measures}
\label{sec:gen_measures}

So far, we presented convergence results to a Wasserstein barycenter for classes of discrete (Theorem~\ref{th:discr-gen} and Corollary~\ref{co:discr-uni}) and absolutely continuous (Theorem~\ref{th:cont-gen} and Corollary~\ref{co:discr-gen-uni}) measures. 
In these cases an optimal transport map exists between any two agents' measures at every time. 
Now we analyze 
the case for 
general measures in $\PPd$, which includes cases where there may not exist an optimal transport map between two initial measures or where there could exist a mix of discrete and absolutely continuous initial measures.
%
%

\begin{theorem}[Consensus result for general measures in $\PPd$]
\label{th:gen-conv}
Consider the PaWBar algorithm with an underlying graph $G$ which is either
\begin{enumerate}
\item\label{it:gen-1} a cycle graph for the directed case,  or
\item\label{it:gen-2} a line graph for the symmetric case; 
\end{enumerate}
and with the agents having initial measures $\mu_{i,0}\in\PPd$, $i\in V$. Then, with probability one, for any $i\in V$, 
\begin{equation}
\label{eq:conv-gen}
W_2(\mu_i(t),\mu_{\infty})\to 0\text{ as }t\to\infty,
\end{equation}
where $\mu_{\infty}\in\PPd$ is a random measure whose possible values may depend on the realization of the edge selection process. If $\mu_{i,0}=\mu_{j,0}$ for any $i,j\in V$, then $\mu_{\infty}=\mu_{i,0}$ with probability one.
%
\end{theorem}

\begin{remark}[Open problem: Theorem~\ref{th:gen-conv} and Wasserstein barycenters]
\label{re:1}
%
The characterization of the consensus value in our theorem  
does not state any sufficient condition under which the converged consensus random measure is a Wasserstein barycenter or not: 
this is an open problem for further research.
\end{remark}

For the particular case of general measures in $\PP(\R)$, consensus is guaranteed under the same general conditions for the underlying communication graph as in the previous results of our paper. 
This follows
from a direct application of~\cite[Theorem~1]{ANB-AD:21}.

\begin{theorem}[Consensus result for general measures in $\PP(\R)$]
\label{th:gen-conva}
Consider the agents having initial measures $\mu_{i,0}\in\PP(\R)$, $i\in V$. Then, for the directed, resp. symmetric, PaWBar algorithm under a strongly connected digraph, resp. connected undirected graph,
\begin{equation}
\label{eq:conv-gena}
W_2(\mu_i(t),\mu_{\infty})\to 0\text{ as }t\to\infty,
\end{equation}
where $\mu_{\infty}\in\PP(\R)$ is a random measure whose possible values may depend on the realization of the edge selection process.
\end{theorem}

\section{Proofs of results in Section~\ref{sec:prop-alg}}
\label{sec:proof}

\subsection{Proofs of results in Subsection~\ref{sub:discr}}

\begin{proof}[Proof sketch of Theorem~\ref{th:discr-gen}]
Since measures $\mu_{i,0}$ and $\mu_{j,0}$, $i,j\in V$, are discrete uniform, we have \\
$W_2^2(\mu_{i,0},\mu_{j,0})=\min_{\sigma\in\Sigma_N}\frac{1}{N}\sum_{k=0}^N\norm{x^i_k-x^j_{\sigma(k)}}_2^2$~\cite{CV:03,JR-GP-JD-MB:12}, with $\Sigma_N$ being the set of all possible permutations of the elements in $\until{N}$; i.e., any permutation map $\sigma\in\Sigma_N$ is a bijective function $\sigma:\until{N}\to\until{N}$. 
Then, the displacement interpolation between any of the initial measures provides discrete measures. 
Now, consider $\sigma\in\Sigma_N$ from solving the optimal transport problem from $\mu_{i,0}$ to $\mu_{j,0}$, i.e., $\gamopt(\mu_{i,0},\mu_{j,0})=\frac{1}{N}\sum_{k=0}^N\delta_{(x^i_k,x^j_{\sigma(k)})}$. 
Consider two arbitrary points $(x^i_{k_1},x^j_{\sigma(k_1)}),(x^i_{k_2},x^j_{\sigma(k_2)})\in\supp(\gamopt(\mu_{i,0},\mu_{j,0}))$. Then, for any $a_{ij}\in(0,1)$, the displacement interpolation implies the existence of some $z_{k_1}(a_{ij}),z_{k_2}(a_{ij})\in\supp((\pi_{a_{ij}})_{\#}\mu_{i,0})$, such that $z_{k_1}(a_{ij})=(1-a_{ij})x^i_{k_1}+a_{ij}x^j_{\sigma(k_1)}$ and $z_{k_2}(a_{ij})=(1-a_{ij})x^i_{k_2}+a_{ij}x^j_{\sigma(k_2)}$. 
Now, since the optimal transport plan $\gamopt(\mu_{i,0},\mu_{j,0})$ 
has cyclically monotone support~\cite[Section~2.3]{CV:03}, we can follow the treatment in~\cite[Chapter~8]{CV:09} and conclude that 
there exists no $a_{ij}\in(0,1)$ such that $z^{k_1}(a_{ij})=z^{k_2}(a_{ij})$. As a consequence, $\supp((\pi_{a_{ij}})_{\#}\mu_{i,0})$ has $N$ (different) elements for any possible edge weight $a_{ij}$, i.e., $(\pi_{a_{ij}})_{\#}\mu_{i,0}$ is a discrete uniform measure. 
It is easy to show by induction that in either the directed or symmetric PaWBar algorithm, $\mu_i(t)$ is a discrete uniform distribution for any $i\in V$ and time $t$ with probability one.

We now introduce some notation. Given $A\in\R^{m\times m}$, let $\diag^{i,k}(A)\in\R^{km\times km}$ be the $k\times k$ block-diagonal matrix such that its $i$th block has the matrix $A$ and the rest of blocks are $I_m$. 
Given $A,B\in\R^{m\times m}$, let $\diag^{ij,k}(A,B)\in\R^{km\times km}$ be the $k\times k$ block-diagonal matrix such that its $i$th and $j$th blocks are the matrices $A$ and $B$ respectively, and the rest of blocks are $I_m$.
Let $\bx^i(t)\in\R^{Nd}$ be a vector stacking the elements of $\supp(\mu_i(t))$, which we call the \emph{support vector}. Note that since the measures are discrete uniform at every time $t$ (with probability one), the order of the elements $x^i_{k}(t)\in\R^{d}$, $k\in\until{N}$, in the vector $\bx^i(t)$ can be arbitrary; but for convenience we denote it as $\bx^i(t)=(x^i_{1}(t),\cdots,x^i_{N}(t))^\top$. 
For any $i,j\in V$ and time $t$, let $\sigma_{ij,t}\in\Sigma_N$ be an optimal transport map from $\mu_i(t)$ to $\mu_j(t)$; and let $\sigma_{ii,t}(k)=k$ and $\sigma_{ji,t}=\sigma_{ij,t}^{-1}$.

We now focus on proving statement~\ref{it-dis-gen-dir}. Assume $(i,j)\in E$ is selected at time $t$. Then, 
$x^i_{k}(t+1)=(1-a_{ij})x^i_{k}(t)+a_{ij}x^j_{\sigma_{ij,t}(k)}(t)$, 
$k\in\until{N}$, i.e., 
\begin{equation}
\label{eq:norm-updt}
\bx^i(t+1)=(1-a_{ij})\bx^i(t)+a_{ij}(P(t)\otimes I_d)\bx^j(t)
\end{equation}
with the permutation matrix $P(t)\otimes I_d$ defined by the permutation
matrix $P(t)\in\{0,1\}^{N\times N}$ whose $k$th row is
$e_{\sigma_{ij,t}(k)}^\top$. Indeed, with $Q_{ij,t}=P(t)\otimes I_d$,
\begin{equation*}
W_2^2(\mu_i(t),\mu_j(t))=\frac{1}{N}\norm{\bx^i(t)-Q_{ij,t}\bx^j(t)}_2^2.
\end{equation*}
Now, set $\bx(t)=(\bx^1(t),\dots,\bx^n(t))^\top\in\R^{nNd}$. Then,~\eqref{eq:norm-updt} can also be expressed as
\begin{equation}
\label{eq:norm-updt-tod}
\bx(t+1)=B(t)\bx(t)
\end{equation}
with the row-stochastic matrix $B(t)=\diag^{j,n}(P(t)^\top\otimes I_d)(A(t)\otimes I_{Nd})\diag^{j,n}(P(t)\otimes I_d)$. 

Consider now an initial vector $\bx(0)$ and a 
fixed realization of the edge selection process. 
Now consider $\bx'(0)=\diag(P_1\otimes I_d,\dots,P_n\otimes I_d)\bx(0)$ with arbitrary permutation matrices $P_1,\dots,P_n\in\{0,1\}^{N\times N}$. Notice that both $\bx(0)$ and $\bx'(0)$ represent the supports of the same group of measures $\{\mu_{i,0}\}_{i\in V}$ but may be the case that $\bx(0)\neq\bx'(0)$. We claim that 
\begin{equation}
\label{eq:cl-1}
\bx'(t)=\diag(P_1\otimes I_d,\dots,P_n\otimes I_d)\bx(t)\text{ for any time }t.
\end{equation}
To prove this claim, first recall that we have a fixed realization of the edge selection process. Assume $(i,j)\in E$ is selected at time $t=0$ and obtain $\bx(1)=B(0)\bx(0)$. Likewise, $\bx'(1)=B'(0)\bx'(0)$, with $B'(0)=\diag^{j,n}(P'(0)^\top\otimes I_d)(A(0)\otimes I_{Nd})\diag^{j,n}(P'(0)\otimes I_d)$, is the update that results if the algorithm starts with initial vector $\bx'(0)$. Then,
$$
P'(0)\otimes I_d=(P_i\otimes I_d)(P(0)\otimes I_d)(P_j^\top\otimes I_d)=(P_iP(0)P_j^\top)\otimes I_d.
$$ 
After some algebraic work, we obtain $B'(0)=\diag^{ij,n}(P_i,P_j)B(0)\diag^{ij,n}(P_i^\top,P_j^\top)$,
and so 
\begin{align*}
\bx'(1)&=\diag^{ij,n}(P_i,P_j)B(0)\diag^{ij,n}(P_i^\top,P_j^\top)\bx'(0)\\
&=\diag^{ij,n}(P_i,P_j)B(0)\diag^{ij,n}(P_i^\top,P_j^\top)\diag(P_1,\dots,P_n)\bx(0)\\
&=\diag(P_1,\dots,P_n)B(0)\bx(0)=\diag(P_1,\dots,P_n)\bx(1).
\end{align*}
Finally~\eqref{eq:cl-1} is easily proved by induction, and the claim is proved.

Now, let us make the following claim:
\begin{enumerate}[label=(i.\alph*)]
\item\label{aux:claim_3} for any $i^*,j^*\in V$, $i^*\neq j^*$, $\epsilon>0$ and time $t$, the event 
``$W_2(\mu_{i^*}(t+T),\mu_{j^*}(t+T))<\epsilon$ 
for some finite $T>0$" has positive probability. 
\setcounter{saveenum}{\value{enumi}}
\end{enumerate}
We prove the claim. Define
$d(i,j,\sigma^i,\sigma^j,t_i,t_j):=\Big(\sum^N_{k=1}\frac{1}{N}
\bignorm{x^i_{\sigma^i(k)}(t_i)-x^j_{\sigma^j(k)}(t_j)}^2_2\Big)^{\frac{1}{2}}$,
for $i,j\in V$, $\sigma^i,\sigma^j\in\Sigma_N$.
Consider any $i^*,j^*\in V$ and $\epsilon>0$. 
Since $G$ is strongly connected, there exists a shortest directed path $\Pa_{i^*\to j^*}=((i^*,\ell_{1}),\dots, (\ell_{L-1},j^*))$ from $i^*$ to $j^*$ of some length $L$. 
Now, pick positive numbers $\epsilon_1,\dots,\epsilon_L$ such that $\sum^L_{i=1}\epsilon_i<\epsilon$. Consider any time $t$. Then, we can first select $T_1$ times the edge $(\ell_{L-1},j^*)$ so that 
\begin{equation*}
  d(\ell_{L-1},j^*,\sigma^{-1}_{\ell_{L-1}j^*,t+T_1},I\!d,t+T_1,t)=(1-a_{\ell_{L-1}j^*})^{T_1}d(\ell_{L-1},j^*,\sigma^{-1}_{\ell_{L-1}j^*,t},I\!d,t,t)<\epsilon_{L}.
\end{equation*}
Then, we can select $T_2$ times the edge $(\ell_{L-2},\ell_{L-1})$ so that 
\begin{equation*}
  d(\ell_{L-2},\ell_{L-1},\sigma^{-1}_{\ell_{L-2}\ell_{L-1},t+T_1+T_2}\circ\sigma^{-1}_{\ell_{L-1}j^*,t+T_1},\sigma^{-1}_{\ell_{L-1}j^*,t+T_1},t+T_1+T_2,t+T_1)<\epsilon_{L-1},
\end{equation*}
and we can continue like this until finally selecting $T_{L}$ times the
edge $(i^*,\ell_{1})$ such that \\
$d(i^*,\ell_{1},\sigma,\sigma^{-1}_{\ell_1\ell_2,t+\sum^{L-1}_{i=0}
  T_i}\circ\dots\circ\sigma^{-1}_{\ell_{L-1}j^*,t},t+T,t+\sum^{L-1}_{i=1}
T_i)<\epsilon_{1}$ with
$\sigma=\sigma^{-1}_{i^*\ell_1,t+T}\circ\dots\circ\sigma^{-1}_{\ell_{L-1}j^*,t}$
and $T=\sum^{L}_{i=1} T_i$.  Then,
\begin{align*}
W_2(\mu_{i^*}(t+T),\mu_{j^*}(t+T))\leq(\sum_{k=1}^N\frac{1}{N}\norm{x^{i^*}_{\sigma(k)}(t+T)-x^{j^*}_{k}(t+T)}_2^2)^{\frac{1}{2}}<\sum_{i=1}^{L} \epsilon_i< \epsilon,
\end{align*}
where the first inequality follows by definition of the Wasserstein distance, and the second inequality from both the triangle inequality and the fact that $x^{j^*}_{k}(t+T)=x^{j^*}_{k}(t)$, $x^{\ell_{L-1}}_k(t+T)=x^{\ell_{L-1}}_k(t+T_1)$, $\dots$, $x^{\ell_1}_{k}(t+T)=x^{\ell_1}_{k}(t+\sum^{L-1}_{i=1}T_i)$. Moreover, our construction implies, for any $p\in V$ in the path $\Pa_{i^*\to j^*}$,
\begin{equation}
\label{eq:path}
W_2(\mu_{p}(t+T),\mu_{j^*}(t+T))<\epsilon.
\end{equation}
Now, consider any $m\in V$ and construct a directed acyclic subgraph $G'=(V,E')$, $E'\subset E$, of $G$ as follows: $m$ is the unique node with zero out-degree (i.e, $(m,i)\notin E'$ for any $i\in V$) and there exists a unique directed path from any node $i\in V\setminus\{m\}$ to $m$. Such subgraph $G'$ exists because $G$ is strongly connected. Consider any $\epsilon>0$ and time $t$. Then, the selection process just described above can make all nodes $\bar{m}$ with zero in-degree in $G'$ (i.e., any $\bar{m}\in V$ such that $(i,\bar{m})\notin E'$ for any $i\in V\setminus\{\bar{m}\}$) 
satisfy $W_2(\mu_{\bar{m}}(t+T),\mu_m(t+T))=W_2(\mu_{\bar{m}}(t+T),\mu_m(t))<\frac{\epsilon}{2}$ for some $T$. 
Then, as a consequence of~\eqref{eq:path}, $W_2(\mu_i(t+T),\mu_j(t+T))<\frac{\epsilon}{2}$ for any $i\in V$, and the triangle inequality then implies $W_2(\mu_i(t+T),\mu_j(t+T))<\epsilon$ for any $j\in V$. 
Finally, for any $i,j\in V$, the event ``$W_2(\mu_i(t+T),\mu_j(t+T))<\epsilon$ for some $T>0$" has a positive probability to occur at any time $t$ because any selection of a finite sequence of edges has positive probability to occur at any time $t$. This finishes the proof of claim~\ref{aux:claim_3}.

Now, the event in~\ref{aux:claim_3}, due to its persistence, will eventually happen with probability one. Assume it happens at time $t$. 
Then we claim that $\epsilon$ in this event in this event could have been chosen sufficiently small so that, for any time $t'\geq t$ and any $i,j,p\in V$,
\begin{enumerate}[label=(i.\alph*)]\setcounter{enumi}{\value{saveenum}}
\item\label{aux:claim_otr_1} $\sigma_{ij,t'}=\sigma_{ij,t}$, and
\item\label{aux:claim_otr_22} $\sigma_{ij,t'}=\sigma_{ip,t'}\circ\sigma_{pj,t'}$.
\end{enumerate}
Now we prove the claim. Firstly, note that from~\ref{aux:claim_3} and the fact that the measures are discrete uniform at every time with probability one, we can consider a small enough $\epsilon$ such that for any $i,j\in V$ and any permutation map $\sigma\neq \sigma_{ij,t}$,
\begin{equation}
\label{eq:cond-clos}
\begin{split}
&W_2(\mu_i(t),\mu_j(t))
=\Big(\frac{1}{N}\sum_{k=1}^N\norm{x_k^i(t)-x_{\sigma_{ij,t}(k)}^j(t)}_2^2\Big)^{\frac{1}{2}}<\epsilon
%
%
\quad\text{ and }\\
&2\epsilon<\Big(\frac{1}{N}\sum_{k=1}^N\norm{x_k^i(t)-x_{\sigma(k)}^j(t)}_2^2\Big)^{\frac{1}{2}}.
\end{split}
\end{equation} 
Such choice of $\epsilon$ implies that $\sigma_{ip,t}\circ\sigma_{pj,t}=\sigma_{ij,t}$ for any $i,j,p\in V$; otherwise, if $\sigma_{ip,t}\circ\sigma_{pj,t}\neq\sigma_{ij,t}$, then we obtain a contradiction:
\begin{align*}
2\epsilon&<\Big(\frac{1}{N}\sum_{k=1}^N\norm{x_k^i(t)-x_{\sigma_{ip,t}\circ\sigma_{pj,t}(k)}^j(t)}_2^2\Big)^{\frac{1}{2}}=\Big(\frac{1}{N}\sum_{k=1}^N\norm{x^i_{\sigma_{pi,t}(k)}(t)-x_{\sigma_{pj,t}(k)}^j(t)}_2^2\Big)^{\frac{1}{2}}\\
&\leq\Big(\frac{1}{N}\sum_{k=1}^N\norm{x^i_{\sigma_{pi,t}(k)}(t)-x_k^p(t)}_2^2\Big)^{\frac{1}{2}}+\Big(\frac{1}{N}\sum_{k=1}^N\norm{x^p_k(t)-x_{\sigma_{pj,t}(k)}^j(t)}_2^2\Big)^{\frac{1}{2}}<2\epsilon.
\end{align*}
We just proved that~\ref{aux:claim_otr_22} holds for $t'=t$. Note that~\ref{aux:claim_otr_1} for $t'=t$ is trivial.
Now, assume any $(i^*,j^*)\in E$ is selected at time $t$. Then, for any $j\in V\setminus\{i^*,j^*\}$, using the identity $Q_{j^*j,t}=Q_{j^*i^*,t}Q_{i^*j,t}$ from~\ref{aux:claim_otr_22} for $t'=t$ implies
\begin{align}
 & \norm{\bx^{i^*}(t+1)-Q_{i^*j,t}\bx^j(t+1)}_2\\
  & \quad \leq(1-a_{i^*j^*})\norm{\bx^{i^*}(t)-Q_{i^*j,t}\bx^j(t)}_2
  +a_{i^*j^*}\norm{Q_{i^*j^*,t}\bx^{j^*}(t)-Q_{i^*j,t}\bx^j(t)}_2 \nonumber \\
  & \quad =(1-a_{i^*j^*})\norm{\bx^{i^*}(t)-Q_{i^*j,t}\bx^j(t)}_2+a_{i^*j^*}\norm{\bx^{j^*}(t)-Q_{j^*j,t}\bx^j(t)}_2 \nonumber \\
  &\quad <(1-a_{i^*j^*})\epsilon\sqrt{N}+a_{i^*j^*}\epsilon\sqrt{N}
  =\epsilon\sqrt{N}; \nonumber
\end{align}
likewise, we immediately obtain $\frac{1}{\sqrt{N}}\norm{\bx^i(t+1)-Q_{ij,t}\bx^j(t+1)}_2<\epsilon$ for any $i\in V\setminus\{i^*,j\}$, and $\frac{1}{\sqrt{N}}\norm{\bx^{i^*}(t+1)-Q_{i^*j^*,t}\bx^{j^*}(t+1)}_2<(1-a_{i^*j^*})\epsilon<\epsilon$. In summary, \\$\frac{1}{\sqrt{N}}\norm{\bx^{i}(t+1)-Q_{ij,t}\bx^{j}(t+1)}_2<\epsilon$ for any $i,j\in V$, which implies $\sigma_{ij,t+1}=\sigma_{ij,t}$ for any $i,j\in V$; i.e.,~\ref{aux:claim_otr_1} holds for $t'=t+1$. Now, to prove claim~\ref{aux:claim_otr_22} holds for $t'=t+1$, we must first prove that~\eqref{eq:cond-clos} holds for time $t+1$. 

Set $\by^1(t):=\bx^1(t),\by^2(t)=Q_{12,t}\bx^2(t),\dots,\by^n(t)=Q_{1n,t}\bx^n(t)$ (this labeling is arbitrary and 
any other $i\in V\setminus\{1\}$ could have been chosen to define $Q_{i1},\dots,Q_{in}$)
and $\by^i(t)=(y^i_1(t),\dots,y^i_N(t))^\top$, $y^i_1(t)\in\R^d$, $i\in\until{n}$. For any $k\in\until{N}$, let $\mL_k(t)$ be the convex hull of the set $\{y^1_k(t),\dots,y^n_k(t)\}$. For any $p,q\in\until{N}$, define the distance between $\mL_p(t)$ and $\mL_q(t)$ as $d_{pq}(t)=\inf_{w_1\in\mL_p(t),w_2\in\mL_q(t)}\norm{w_1-w_2}_2$. 
Assuming that $(i^*,j^*)\in E$ is selected at time $t$, our result~\ref{aux:claim_otr_1} for $t'=t+1$ and~\eqref{eq:norm-updt} imply that $y^{i^*}_k(t+1)=(1-a_{i^*j^*})y^{i^*}_k(t)+a_{i^*j^*}y^{j^*}_k\in\mL_k(t)$, $k\in\until{N}$. Obviously, for any $j\in V\setminus\{i^*\}$, $y^j_k(t+1)=y^j_k(t)\in\mL_k(t)$, $k\in\until{N}$. Then, $\mL_i(t+1)\subseteq\mL_i(t)$ for any $i\in\until{N}$, and thus $d_{pq}(t)\leq d_{pq}(t+1)$ for any $p,q\in\until{N}$. 
Now, for any $i,j\in V$, time $\tau\geq t$, and permutation map $\sigma\neq I\!d$, we have $d_{k\sigma(k)}(\tau)\leq \norm{y^i_k(\tau)-y^j_{\sigma(k)}(\tau)}_2\Longrightarrow (\sum_{k=1}^Nd^2_{k\sigma(k)}(\tau))^{\frac{1}{2}}\leq (\sum_{k=1}^N\norm{y^i_k(\tau)-y^j_{\sigma(k)}(\tau)}_2^2)^{\frac{1}{2}}$. Now, we need to consider two cases. In the first case we consider \\$2\epsilon<\min_{\bar{\sigma}\in\Sigma_N,\bar{\sigma}\neq I\!d}(\frac{1}{N}\sum_{k=1}^Nd^2_{k\bar{\sigma}(k)}(t))^{\frac{1}{2}}$. Then, $$2\epsilon<(\frac{1}{N}\sum_{k=1}^Nd^2_{k\sigma(k)}(t+1))^{\frac{1}{2}}\leq(\frac{1}{N}\sum_{k=1}^N\norm{y_k^i(t+1)-y_{\sigma(k)}^j(t+1)}_2^2)^{\frac{1}{2}}$$ for any permutation map $\sigma\neq I\!d$; and~\ref{aux:claim_otr_1} for $t'=t+1$ and~\ref{aux:claim_otr_22} for $t'=t$ imply $2\epsilon<(\frac{1}{N}\sum_{k=1}^N\norm{x_k^i(t+1)-x_{\sigma'(k)}^j(t+1)}_2^2)^{\frac{1}{2}}$ with $\sigma'=\sigma_{i1,t+1}\circ\sigma\circ\sigma_{1j,t+1}\neq \sigma_{ij,t+1}$. Thus,~\eqref{eq:cond-clos} holds for time $t+1$ in this first case.  
Now, we consider the second case $2\epsilon\geq\min_{\bar{\sigma}\in\Sigma_N,\bar{\sigma}\neq I\!d}(\frac{1}{N}\sum_{k=1}^Nd^2_{k\bar{\sigma}(k)}(t))^{\frac{1}{2}}$. Then, due to $G$ being strongly connected and $\{d_{pq}(\tau)\}_{\tau\geq t}$ being a nondecreasing sequence for any $p,q\in\until{N}$, we can follow the proof of result~\ref{aux:claim_3} and arbitrarily reduce the diameter of the set $\mL_k$ for any $k\in\until{N}$ at some future time $\bar{t}$, i.e., $\mL(\bar{t})\subset\mL(t)$. This diameter reduction can be chosen such that $d_{ij}(\bar{t})>d_{ij}(t)$ for any $i,j\in\until{N}$, and this increase on the distances between sets 
can be done so that 
$2\epsilon'<\min_{\bar{\sigma}\in\Sigma_N,\bar{\sigma}\neq I\!d}(\frac{1}{N}\sum_{k=1}^Nd^2_{k\bar{\sigma}(k)}(\bar{t}))^{\frac{1}{2}}$ for some $0<\epsilon'<\epsilon$. In other words, we are in the first case at time $\bar{t}$. After this change, we will never be in the second case again for any time after $\bar{t}$ with probability one. In summary, we just proved the conditions in equation~\eqref{eq:cond-clos} can be made to hold for time $t+1$, and so \ref{aux:claim_otr_22} holds for $t'=t+1$.

Now, assume results~\ref{aux:claim_otr_1} and~\ref{aux:claim_otr_22} hold for time $t'=\tau\geq t$, and~\eqref{eq:cond-clos} holds for time $\tau$. Following the proof just presented above, we easily establish that~\ref{aux:claim_otr_1} and~\ref{aux:claim_otr_22} hold for $t'=\tau+1$ and that~\eqref{eq:cond-clos} holds for time $\tau+1$. Then, by induction, we proved our initial claim about~\ref{aux:claim_otr_1} and~\ref{aux:claim_otr_22}.

We now, notice that $\epsilon$ in $\max_{i,j\in V}W_2(\mu_{i,0},\mu_{j,0})<\epsilon$ can be made sufficiently small so that~\ref{aux:claim_otr_1} and~\ref{aux:claim_otr_22} hold, in which case~\ref{aux:claim_3} is satisfied at the beginning of time, i.e., with $t=T=0$. Therefore, in general, from results~\ref{aux:claim_3},~\ref{aux:claim_otr_1} and~\ref{aux:claim_otr_22}, there exists some (possibly) random time $\bar{T}\geq 0$ such that, with probability one: for any time $t\geq \bar{T}$ and any $i,j,p\in V$, $\sigma_{ij,t}=\sigma_{ij,\bar{T}}$ and $\sigma_{ij,t}=\sigma_{ip,t}\circ\sigma_{pj,t}$.
Let us consider a fixed realization of the edge selection process, and then consider such time $\bar{T}$, which is now a deterministic function of $\bx(0)$.
Without loss of generality, as a consequence of~\eqref{eq:cl-1}, we can assume we started the algorithm with the initial support vectors $\{Q_{1i,\bar{T}}\bx^i(0)\}_{i\in V}$ at time $t=0$. Then, it is easy to prove that $B(t)=A(t)\otimes I_{Nd}$ (as in~\eqref{eq:norm-updt-tod}) for any $t\geq \bar{T}$, i.e., $B(t)$ has an associated permutation matrix $P(t)=I_N$. 
Then, Proposition~\ref{prop:conv_A} let us conclude that 
$\lim_{t\to\infty}\prod^{t}_{\tau=\bar{T}}B(\tau)=(\vect{1}_n\lambda^\top)\otimes I_{Nd}$ for some convex vector $\lambda$. Thus, $\bx^i(\infty)=\sum^n_{j=1}\lambda_j\bx^j(\bar{T})$, $i\in V$, which is the support vector of the final consensus measure $\mu_\infty$.

It remains to prove that $\mu_{\infty}$ corresponds to a Wasserstein
barycenter. Let us formulate the Wasserstein barycenter problem
$\min_{\nu\in\PPd}\sum_{i=1}^n\lambda_iW_2(\nu,\mu_i(\bar{T}))^2$. Since the
measures $\{\mu_i(\bar{T})\}_{i\in V}$ have finite support, any barycenter is a
discrete measure with finite support~\cite{EA-SB-JM:16}.  Moreover, since all the
measures are uniform, we can consider a minimizer with a discrete uniform
distribution. We now prove that $\mu_{\infty}$ is such a
minimizer. Firstly, by construction and the fact that~\eqref{eq:cond-clos} holds for $t\geq \bar{T}$, we have
\begin{align*}
W_2(\mu_{\infty},\mu_i(\bar{T}))^2&= \frac{1}{N}\sum^N_{j=1}\norm{\sum_{k=1}^n\lambda_kx^k_j(\bar{T})-x^i_j(\bar{T})}_2^2\\
  &< 2\epsilon<\frac{1}{N}\sum^N_{j=1}\norm{\sum^n_{k=1}\lambda_k
    x^k_{\sigma^k(j)}(\bar{T})-x^i_{\sigma^i(j)}(\bar{T})}_2^2
\end{align*}
for any $\sigma^i\in\Sigma_N$, $\sigma_i\neq I\! d$, $i\in V$. 
Then, 
\begin{align*}
  \sum_{i=1}^n\lambda_i W_2(\mu_{\infty},\mu_i(\bar{T}))^2&<
  \frac{1}{N}\sum^n_{i=1}\lambda_i\sum^N_{j=1}
  \norm{\sum^n_{k=1}\lambda_k x^k_{\sigma^k(j)}(\bar{T})-x^i_{\sigma^i(j)}(\bar{T})}_2^2\\
   &\leq \frac{1}{N}\sum^n_{i=1}\lambda_i\sum^N_{j=1}
  \norm{y_i-x^i_{\sigma^i(j)}(\bar{T})}_2^2
\end{align*}
for any $y=(y_1,\dots,y_n)^\top\in\R^{Nd}$. The last inequality of the previous expression is proved by 
treating the last term as an objective function to minimize with respect to $y$ 
using first optimality conditions to minimize such differentiable and strictly convex function (i.e., by setting the gradient with respect to $y$ equal to the zero vector and solving for $y$). Now, take $y_i\neq y_j\in\R^d$ for any $i,j\in \until{N}$, define the discrete uniform measure $\nu=\frac{1}{N}\sum^N_{j=1}\delta_{y_j}$, $y_j\in\R^d$; and let $\bar{\sigma}^i$, $i\in V$, be such that $W_2(\nu,\mu_i(\bar{T}))^2=\frac{1}{N}\sum^n_{i=1}\lambda_i\sum^N_{j=1}\norm{y_i-x^i_{\bar{\sigma}^i(j)}(\bar{T})}_2^2$. Then, our recent analysis implies that: (1) if there exists $i\in V$ such that $\bar{\sigma}_i\neq I\! d$, then $\sum_{i=1}^n\lambda_iW_2(\mu_{\infty},\mu_i(\bar{T}))^2<\sum_{i=1}^n\lambda_i W_2(\nu,\mu_i(\bar{T}))^2$; \\
(2) if $\bar{\sigma}_i= I\! d$ for all $i\in V$, then $\sum_{i=1}^n\lambda_i
W_2(\mu_{\infty},\mu_i(\bar{T}))^2\leq\sum_{i=1}^n\lambda_i
W_2(\nu,\mu_i(\bar{T}))^2$. \\ Given the generality of $\nu$, cases (1) and (2)
together imply that $\mu_{\infty}$ is a Wasserstein barycenter.

Finally, all of our previous results hold with probability one because we considered an arbitrary realization of the edge selection process for our analysis (note that $\lambda$ now becomes a random convex vector). This concludes the proof of statement~\ref{it-dis-gen-dir}.

We now focus on proving statement~\ref{it-dis-gen-sym}. Assume $\{i,j\}\in E$ is selected at time $t$. Without loss of generality, the update of the PaWBar algorithm can be set as 
$x^i_{k}(t+1)=\frac{1}{2}x^i_{k}(t)+\frac{1}{2}x^j_{\sigma_{ij,t}(k)}(t)$ and $x^j_{k}(t+1)=\frac{1}{2}x^j_{k}(t)+\frac{1}{2}x^i_{\sigma_{ji,t}(k)}(t)$, 
$k\in\until{N}$; i.e., 
\begin{equation}
\label{eq:norm-updt-s}
\begin{aligned}
\bx^i(t+1)&=\frac{1}{2}\bx^i(t)+\frac{1}{2}(P(t)\otimes I_d)\bx^j(t),\\
\bx^j(t+1)&=\frac{1}{2}\bx^j(t)+\frac{1}{2}(P(t)^\top\otimes I_d)\bx^i(t)
\end{aligned}
\end{equation}
recalling that the permutation matrix $P(t)\in\{0,1\}^{N\times N}$ has $e_{\sigma_{ij,t}(k)}^\top$ as its $k$th row. Note that~\eqref{eq:norm-updt-s} can also be expressed as $\bx(t+1)=C(t)\bx(t)$, with matrix $C(t)=\diag^{i,n}(P(t)\otimes I_d)(A(t)\otimes I_{Nd})\diag^{i,n}(P^\top(t)\otimes I_d)$.

We make the following claim:
\begin{enumerate}[label=(ii.\alph*)]
\item\label{aux:claim_otr_2_1} for any $i^*,j^*\in V$, $i^*\neq j^*$, $\epsilon>0$ and time $t$, the event  
``$W_2(\mu_{i^*}(t+T),\mu_{j^*}(t+T))<\epsilon$ 
for some finite $T>0$" has positive probability. 
\end{enumerate}
Now, we prove the claim. Let us fix a spanning tree $G'$ of $G$. For any $i,j\in V$, let $\mathcal{P}_{i-j}$ denote the unique path between $i$ and $j$ in $G'$. Let 
$$
U(t)=\max_{i,j\in V}
\sum_{\{p,q\}\in\Pa_{i-j}}W_2(\mu_p(t),\mu_q(t)).
$$
Let $\{k,\ell\}\in\arg U(t)$ and $\Pa_{k-l}=(\{k,p_1\},\dots,\{p_{L-1},\ell\})$, i.e., edge $\{k,p_1\}$ is followed by $\{p_1,p_2\}$ and so on until $\{p_{L-1},\ell\}$. \emph{Case 1)} $W_2(\mu_k(t),\mu_{p_1}(t))\neq 0$. 
For simplicity we also assume $W_2(\mu_i(t),\mu_j(t))\neq 0$ for any $\{i,j\}\in\Pa_{k-\ell}$; otherwise, if there exists $\{i^*,j^*\}\in\Pa_{k-\ell}$ such that $W_2(\mu_{i^*}(t),\mu_{j^*}(t))=0$, we would need to use a similar analysis to Case 2) which will be treated later.  
Select $\{k,p_1\}$ at time $t$. If $\Pa_{k-\ell}$ contains only one element, then $p_1=\ell$ and $W_2(\mu_k(t+1),\mu_\ell(t+1))=0<U(t)=W_2(\mu_k(t),\mu_\ell(t))$. Now, consider $\Pa_{k-\ell}$ contains two or more elements. 
Set $\mathcal{U}(t)=\sum_{\{p,q\}\in\Pa_{k-\ell}\setminus\{\{k,p_1\},\{p_1,p_2\}\}}W_2(\mu_p(t),\mu_q(t))$ (with $p_2=\ell$ and $\mathcal{U}(t)=0$ if $\Pa_{k-\ell}$ only has two elements). Then
\begin{align*}
  \sum_{\{p,q\}\in\Pa_{k-\ell}} & W_2(\mu_p(t+1),\mu_q(t+1))=\mathcal{U}(t)
  +\frac{1}{\sqrt{N}}\norm{\bx^{p_1}(t+1)-Q_{p_1p_2,t+1}\bx^{p_2}(t)}_2\\
  &\leq\mathcal{U}(t)
  +\frac{1}{\sqrt{N}}\norm{\bx^{p_1}(t+1)-Q_{p_1p_2,t}\bx^{p_2}(t)}_2\\
  &\leq\mathcal{U}(t)+
  \frac{1}{2}W_2(\mu_{p_1}(t),\mu_{p_2}(t)) + 
  \frac{1}{2\sqrt{N}}\norm{Q_{p_1k,t}\bx^{k}(t)-Q_{p_1p_2,t}\bx^{p_2}(t)}_2\\
  &\leq\mathcal{U}(t)
  +\frac{1}{2}W_2(\mu_{p_1}(t),\mu_{p_2}(t))
  +\frac{1}{2\sqrt{N}}\norm{Q_{p_1k,t}\bx^{k}(t)-\bx^{p_1}(t)}_2\\
  &\quad +\frac{1}{2\sqrt{N}}\norm{\bx^{p_1}(t)-Q_{p_1p_2,t}\bx^{p_2}(t)}_2\\
  &=\mathcal{U}(t)+W_2(\mu_{p_1}(t),\mu_{p_2}(t))+\frac{1}{2}W_2(\mu_{k}(t),\mu_{p_1}(t)),
\end{align*}
and so $\sum_{\{p,q\}\in\Pa_{k-\ell}}W_2(\mu_p(t+1),\mu_q(t+1))<U(t)$. Therefore, for any length of $\Pa_{k-\ell}$, if $U(t+1)\leq \sum_{\{p,q\}\in\Pa_{k- \ell}}W_2(\mu_p(t+1),\mu_q(t+1))$, then $U(t+1)<U(t)$. If $U(t+1)>\sum_{\{p,q\}\in\Pa_{k- \ell}}W_2(\mu_p(t+1),\mu_q(t+1))$, then we can choose $\{\bar{k},\bar{\ell}\}\in\arg U(t+1)$ and, using the analysis just presented, obtain $\sum_{\{p,q\}\in\Pa_{\bar{k}-\bar{\ell}}}W_2(\mu_p(t+2),\mu_q(t+2))<U(t+1)$. If this does not imply $U(t+2)<U(t)$, we can keep iterating this procedure until, eventually, obtain $U(t+T)<U(t)$ for some $T>0$. 
\emph{Case 2)} $W_2(\mu_k(t),\mu_{p_1}(t))=0$. In this case, we do not select the edge $\{k,p_1\}$, but we consecutively check the edges along $\Pa_{k-\ell}$ starting from $\{k,p_1\}$ and look for the first $\{i^*,j^*\}\in\Pa_{k-\ell}$ such that $W_2(\mu_{i^*}(t),\mu_{j^*}(t))\neq 0$. We select this edge and a similar analysis to Case 1) implies that $\sum_{\{p,q\}\in\Pa_{k\to \ell}}W_2(\mu_p(t+1),\mu_q(t+1))\leq U(t)$. Then, we select the edge previous to $\{i^*,j^*\}$ and continue to successively select the preceding edges until reaching the first edge $\{k,p_1\}$. Once this edge is selected, say at time $\bar{t}$, the proof of case Case 1) let us conclude that $\sum_{\{p,q\}\in\Pa_{k-\ell}}W_2(\mu_p(\bar{t}+1),\mu_q(\bar{t}+1))<U(\bar{t})$, and we can continue the analysis of Case 1) until we have that $U(t+T)<U(\bar{t})\leq U(t)$ for some $T>0$. 
In conclusion, we proved the existence of some finite sequence of selected edges such that $U(t+T)<U(t)$ for some $T>0$. Moreover, we can iterate selections of such sequences to arbitrarily reduce the value of $U(t)$ after some finite time. Finally, claim~\ref{aux:claim_otr_2_1} follows from the fact that $\max_{i,j\in V}W_2(\mu_i(t),\mu_j(t))\leq U(t)$ and that any finite sequence of edges has a positive probability of being consecutively selected at any time $t$.

We can now follow the same analysis as in the proof of statement~\ref{it-dis-gen-dir} of the theorem -- using result~\ref{aux:claim_otr_2_1} and its proof instead of~\ref{aux:claim_3} -- to conclude that results~\ref{aux:claim_otr_1} and~\ref{aux:claim_otr_22} also hold for the symmetric PaWBar algorithm, after which the proof follows closely the one
%
for statement~\ref{it-dis-gen-dir} again.
\end{proof}

\begin{proof}[Proof of Corollary~\ref{co:discr-uni}]
We follow the notation and proof of Theorem~\ref{th:discr-gen}. Note that the entries of $\bx^i(0)=(x^i_1,\dots,x^i_N)^\top$, $i\in V$, 
are sorted in ascending order. Then, $W^2_2(\mu_{i,0},\mu_{j,0})=\frac{1}{N}\sum_{k=0}^N(x^i_k-x^j_k)^2$ for $i,j\in V$. 
Now, consider the directed PaWBar algorithm and that $(i,j)\in E$ is selected at time $t=0$. Then $\bx^i(1)= (1-a_{ij})\bx^i(0)+a_{ij}\bx^j(0)$ and $\bx^i(1)$ has its entries sorted in ascending order. Then, it is easy to prove by induction that, at every time $t$, $\bx^i(t)$ for any $i\in V$ is sorted in ascending order with probability one. Considering $\bx(t)=(\bx^i(t),\dots,\bx^n(t))^\top\in\R^{nN}$, we have that $\bx(t+1)=(A(t)\otimes I_N)\bx(t)$ and so $\bx(t)=(\prod^t_{i=0}A(i)\otimes \vect{1}_N)\bx(0)$. Then, we conclude the proof for the directed PaWBar algorithm by using Proposition~\ref{prop:conv_A}
and the fact that $\sum_{i=1}^n\lambda_i\bx^i(0)\in\arg\min_{\substack{y\in\R^d\\y_i<\dots<y_n}}\frac{1}{N}\sum^n_{i=1}\lambda_i\sum^N_{k=0}(y_k-x^i_k)^2$ for any convex vector $\lambda\in\R^n$.
The symmetric case is proved similarly.
\end{proof}

\subsection{Proofs of results in Subsection~\ref{sub:cont}}

\begin{proof}[Proof of Theorem~\ref{th:cont-gen}]
Fix any $\gamma\in\{\mu_{i,0}\}_{i\in V}$. Let $\mu(t):=(\mu_1(t),\dots,\mu_n(t))^\top$ and \\$T_\gamma:=(T^{\mu_{1,0}}_\gamma,\dots,T^{\mu_{n,0}}_\gamma)^\top$. We first claim that $\mu(t)=\big(\prod^t_{\tau=0}A(\tau)T_\gamma\big)_\#\gamma$, 
where the push-forward notation $(\cdot)_\#$ is applied element-wise. We will prove this claim by induction.

Assume any $(i,j)\in E$ is selected at $t=0$. Then, 
\begin{equation*}
\begin{aligned}
\mu_i(1)&=((1-a_{ij})I\!d+a_{ij}T^{\mu_{j,0}}_{\mu_{i,0}})_{\#}\mu_{i,0}\\
&=((1-a_{ij})I\!d+a_{ij}T^{\mu_{j,0}}_{\mu_{i,0}})_{\#}(T^{\mu_{i,0}}_\gamma)_\#\gamma\\
&=((1-a_{ij})T^{\mu_{i,0}}_{\gamma}+a_{ij}T^{\mu_{j,0}}_{\mu_{i,0}}\circ T^{\mu_{i,0}}_{\gamma})_{\#}\gamma\\
&=((1-a_{ij})T^{\mu_{i,0}}_{\gamma}+a_{ij}T^{\mu_{j,0}}_{\gamma})_{\#}\gamma,
\end{aligned}
\end{equation*}
where the third equality follows from the property that 
$$(A)_\#(B)_\#\mu=((B)_\#\mu)(A^{-1}(\cdot))=\mu(B^{-1}A^{-1}(\cdot))=\mu((A\circ B)^{-1}(\cdot))=(A\circ B)_\#\mu$$
 for any measure $\mu$ and appropriate measurable maps $A,B$; and the last equality follows from the 
compatible collection. Thus, we have that, with probability one, $\mu(1)=\big(A(0)T_\gamma\big)_\#\gamma$. Again, without loss of generality, let us consider that $(i,j)\in E$ was chosen at $t=0$ 
%
%
%
and analyze all possible updates at $t=1$. Assume some edge $(p,q)\in E$ is chosen. Then, $\mu_p(2)=((1-a_{pq})Id+a_{pq}T^{\mu_q(1)}_{\mu_p(1)})_\#\mu_p(1)$. Now, observe that 
$\mu_i(1)=((1-a_{ij})T^{\mu_{i,0}}_\gamma+a_{ij}T^{\mu_{j,0}}_\gamma)_\#\gamma$ and $\mu_{k}(1)=(T^{\mu_{k,0}}_\gamma)_\#\gamma$, $k\neq i$, and so $\mu_{p}(1)=(\sum^n_{\ell=1}\xi_\ell^p T^{\mu_{\ell,0}}_\gamma)_\#\gamma$ for some nonnegative constants $\{\xi_\ell^p\}_{\ell=1}^n$. Then, $\mu_p(2)=((1-a_{pq})\sum^n_{\ell=1}\xi_\ell^p T^{\mu_{\ell,0}}_\gamma+a_{pq}T^{\mu_q(1)}_{\mu_p(1)}\circ \sum^n_{\ell=1}\xi_\ell^p T^{\mu_{\ell,0}}_\gamma)_\#\gamma$, and so we only need to check $T^{\mu_q(1)}_{\mu_p(1)}\circ\sum^n_{\ell=1}\xi_\ell^p T^{\mu_{\ell,0}}=\sum^n_{\ell=1}\xi_\ell^q T^{\mu_{\ell,0}}$ holds for all different cases of chosen $(p,q)\in E$ to ensure $\mu(2)=(A(1)A(0)T_\gamma)_\#\gamma$. 
In the following cases, we make repeated use of the definition of displacement interpolation and the assumption of closure under interpolation.
\textbf{Case 1 -- $(i,j)\in E$ is chosen again:} $T^{\mu_j(1)}_{\mu_i(1)}\circ((1-a_{ij})T^{\mu_{i,0}}_\gamma+a_{ij}T^{\mu_{j,0}}_\gamma)=T^{\mu_{j,0}}_{\mu_i(1)}\circ((1-a_{ij})Id+a_{ij}T^{\mu_{j,0}}_\gamma\circ T^\gamma_{\mu_{i,0}})\circ T^{\mu_{i,0}}_\gamma
=T^{\mu_{j,0}}_{\mu_i(1)}\circ T^{\mu_i(1)}_{\mu_{i,0}}\circ T^{\mu_{i,0}}_\gamma=T^{\mu_{j,0}}_{\mu_{i,0}}\circ T^{\mu_{i,0}}_\gamma=T^{\mu_{j,0}}_\gamma$.
\textbf{Case 2 -- $(i,k)\in E,k\neq j,$ is chosen:} $T^{\mu_k(1)}_{\mu_i(1)}\circ((1-a_{ij})T^{\mu_{i,0}}_\gamma+a_{ij}T^{\mu_{j,0}}_\gamma)=T^{\mu_{k,0}}_{\mu_i(1)}\circ((1-a_{ij})Id+a_{ij}T^{\mu_{j,0}}_{\mu_{i,0}})\circ T^{\mu_{i,0}}_\gamma=T^{\mu_{k,0}}_{\mu_i(1)}\circ T^{\mu_i(1)}_{\mu_{i,0}}\circ T^{\mu_{i,0}}_\gamma=
T^{\mu_{k,0}}_\gamma$.
\textbf{Case 3 -- $(k,i)\in E$ is chosen:} $T^{\mu_i(1)}_{\mu_k(1)}\circ T^{\mu_{k,0}}_\gamma=
T^{\mu_i(1)}_{\mu_{k,0}}\circ T^{\mu_{k,0}}_\gamma=T^{\mu_i(1)}_\gamma=T^{\mu_i(1)}_{\mu_{i,0}}\circ T^{\mu_{i,0}}_\gamma=((1-a_{ij})Id+a_{ij}T^{\mu_{j,0}}_{\mu_{i,0}})\circ T^{\mu_{i,0}}_\gamma=((1-a_{ij})T^{\mu_{i,0}}_\gamma+a_{ij}T^{\mu_{j,0}}_\gamma)$. 
\textbf{Case 4 -- any other $(p,q)\in E$ chosen:} $T^{\mu_q(1)}_{\mu_p(1)}\circ T^{\mu_{p,0}}_\gamma=T^{\mu_{q,0}}_{\mu_{p,0}}\circ T^{\mu_{p,0}}_\gamma=T^{\mu_{q,0}}_\gamma$.
From cases $1-4$, we conclude that $\mu(2)=\big(A(1)A(0)T_\gamma\big)_\#\gamma$.
Now, let us analyze the update at time $t\geq 2$ with the induction hypothesis that $\mu_k(t)=(\sum^n_{\ell=1}\xi_\ell^k T^{\mu_{\ell,0}}_\gamma)_\#\gamma$ for any $k\in V$, and $T^{\mu_j(t)}_{\mu_i(t)}\circ\big(\sum^n_{\ell=1}\xi^i_{\ell}T^{\mu_{\ell,0}}_\gamma\big)=\sum^n_{\ell=1}\xi^j_{\ell}T^{\mu_{\ell,0}}_\gamma$, with appropriate nonnegative constants $\{\xi^k_\ell\}_{k,\ell}^n$ such that $\mu(t)=(\prod_{\tau=0}^{t-1}A(\tau)T_\gamma)_\#\gamma$. 
If edge $(p,q)\in E$ is selected, $\mu_{p}(t+1)=((1-a_{pq})Id+a_{pq}T^{\mu_q(t)}_{\mu_p(t)})_\#\mu_p(t)=((1-a_{pq})\sum^n_{\ell=1}\xi^p_{\ell}T^{\mu_{\ell,0}}_\gamma+a_{pq}\sum^n_{\ell=1}\xi^q_{\ell}T^{\mu_{\ell,0}}_\gamma)_\#\gamma$, and so 
$\mu(t+1)=(A(t)\prod^{t-1}_{\tau=0}A(\tau)T_\gamma)_\#\gamma$. By induction, this proves the claim.
%

Considering our proved claim, we use Proposition~\ref{prop:conv_A} and obtain 
$$\lim_{t\to\infty}\mu(t)=(\lim_{t\to\infty}\prod_{\tau=0}^{t}A(\tau)T_\gamma)_\#\gamma=(\vect{1}_n\lambda^\top T_\gamma)_\#\gamma$$
for some random convex vector $\lambda=(\lambda_1,\dots,\lambda_n)^\top$ with probability one. This gives the consensus result $\lim_{t\to\infty}\mu_i(t)=(\sum^{n}_{j=1}\lambda_j T^{\mu_{j,0}}_{\gamma})_\#\gamma$.
Finally, we conclude from~\cite[Theorem~3.1.9]{VMP-YZ:20} that the measure $\mu_{\infty}:=\big(\sum^{n}_{j=1}\lambda_j T^{\mu_j(t)}_{\gamma}\big)_{\#}\gamma$ is the unique solution to the barycenter problem with convex vector $\lambda$, i.e., equation~\eqref{eq:dis-cont-infty} is proved. This concludes the proof of statement~\ref{it-cont-gen-dir}. 
Statement~\ref{it-cont-gen-sym} is proved with a similar analysis. 
\end{proof}

\begin{proof}[Proof of Corollary~\ref{co:discr-gen-uni}]
We only focus on proving the results for the directed PaWBar algorithm, since the proofs for the symmetric PaWBar algorithm are very similar and thus omitted. 
Consider any two absolutely continuous measures $\alpha,\beta \in \PP(\R)$. Then, we have 1) $\alpha=(F^{-1}_{\alpha})_{\#}\mathcal{L}$, with $\mathcal{L}$ being the Lebesgue measure on $[0,1]$; and 2) the optimal transport map from $\alpha$ to $\beta$ is the so-called Brenier's map $T^\beta_\alpha = F^{-1}_{\beta}\circ F_{\alpha}$~\cite[Theorem~2.5]{FS:15}. Thus, since any measure obtained from a displacement interpolation is another absolutely continuous measure in $\PP(\R)$, it is straightforward to conclude that the set of all absolutely continuous measures forms a compatible collection which is closed under interpolation. Since Theorem~\ref{th:cont-gen}'s assumption is satisfied, we can fix any $\gamma\in\{\mu_{i,0}\}_{i\in V}$ and replace the Brenier's maps $F^{-1}_{\mu_{i,0}}\circ F_{\gamma}$, $i\in V$,  in the Wasserstein barycenter $\mu_{\infty}$ expression in statement~\ref{it-cont-gen-dir} of Theorem~\ref{th:cont-gen} to conclude the proof.
%
%

Now we consider case~\ref{it:co-2}. 
We first remark that a displacement interpolation between any two initial measures will result in zero-mean multivariate Gaussian variables with a closed form expression for their covariance matrices~\cite{YC-TTG-AT:19}. Thus, we consider a measure $\gamma$ with covariance matrix $\Sigma_\gamma$ resulting from the displacement interpolation with fixed parameter $\lambda\in(0,1)$ between two arbitrary measures $\mu_{i,0}$ and $\mu_{j,0}$ from the initial set of measures. Then, $\Sigma_\gamma=\Sigma_{i,0}^{-1/2}((1-\lambda)\Sigma_{i,0}+\lambda(\Sigma_{i,0}^{1/2}\Sigma_{j,0}\Sigma_{i,0}^{1/2})^{1/2}))^2\Sigma_{i,0}^{-1/2}$ (see ~\cite{YC-TTG-AT:19}), and some algebraic work using the fact that both $\Sigma_{i,0}$ and $\Sigma_{j,0}$ are diagonizable with the orthogonal matrix $U$ let us conclude that $\Sigma_\gamma=U^\top ((1-\lambda)D_{i,0}^{1/2}+\lambda D_{j,0}^{1/2})^2U$. Since a set of zero-mean Gaussian distributions that are diagonizable under the same orthogonal matrix $U$ and which contains the standard Gaussian distribution forms a compatible collection with the linear optimal transport map $T^{\mu_{j,0}}_{\mu_{i,0}}=\Sigma^{1/2}_{j,0}\Sigma^{-1/2}_{i,0}$ (using our notation of the initial set of measures)~\cite[Section~2.3]{VMP-YZ:20}, we just proved that the initial set of measures is closed under interpolation.
Thus, we can use Theorem~\ref{th:cont-gen} to imply the convergence to the Wasserstein barycenter and~\cite[Theorem~2.4]{YC-TTG-AT:19} provides the shown characterization of the barycenter. 
%
\end{proof}

\subsection{Proofs of results in Subsection~\ref{sec:gen_measures}}

\begin{proof}[Proof sketch of Theorem~\ref{th:gen-conv}]
We first consider the directed PaWBar algorithm in case~\ref{it:gen-1}. Consider any $(i,j)\in E$ is selected at time $t$. From the definition of constant-speed geodesics~\cite{FS:15}, it follows that, 
\begin{equation}
\label{eq:geod-1}
\begin{aligned}
&W_2(\mu_i(t+1),\mu_j(t))=(1-a_{ij})W_2(\mu_i(t),\mu_j(t)),\\
&W_2(\mu_i(t+1),\mu_i(t))=a_{ij}W_2(\mu_i(t),\mu_j(t)).
\end{aligned}
\end{equation}
If $(i,j)$ is chosen $\tau$ times consecutively starting at time $t$, then $W_2(\mu_i(t+\tau),\mu_j(t))=(1-a_{ij})^\tau W_2(\mu_i(t),\mu_j(t))$. 

Now, set $$U(t)=\sum_{(i,j)\in E}W_2(\mu_i(t),\mu_j(t)).$$ 
%
Assume any $(i^*,j^*)\in E$ is selected at time $t$, and let $(k^*,i^*)\in E$ (since $G$ is a cycle). Then, setting $\mU(t)=\sum_{(i,j)\in E\setminus{\{(i^*,j^*),(k^*,i^*)\}}}W_2(\mu_i(t),\mu_j(t))$, 
\begin{align*}
U(t+1)
&=W_2(\mu_{i^*}(t+1),\mu_{j^*}(t))+W_2(\mu_{i^*}(t+1),\mu_{k^*}(t))+ \mU(t)\\
&\leq W_2(\mu_{i^*}(t+1),\mu_{j^*}(t))+W_2(\mu_{i^*}(t+1),\mu_{i^*}(t))+W_2(\mu_{i^*}(t),\mu_{k^*}(t))+ \mU(t)\\
&=(1-a_{i^*j^*})W_2(\mu_{i^*}(t),\mu_{j^*}(t))+a_{i^*j^*}W_2(\mu_{i^*}(t),\mu_{j^*}(t))\\
&\quad+W_2(\mu_{i^*}(t),\mu_{k^*}(t))+\mU(t)\\
&=W_2(\mu_{i^*}(t),\mu_{j^*}(t))+W_2(\mu_{i^*}(t),\mu_{k^*}(t))+ \mU(t)=U(t)
\end{align*}
where we used the triangle inequality, and then equation~\eqref{eq:geod-1} for the last equality. Thus, with probability one, $(U(t))_{t\geq 0}$ is a non-increasing sequence uniformly lower bounded by zero, which then implies $U(t)$ converges to some lower bound which we need to prove to be zero.
%
Consider the nontrivial case $U(t)\neq 0$ and again any $(i^*,j^*)\in E$. Since $G$ is a cycle, there is a unique directed path $\Pa_{j^*\to i^*}$ from $j^*$ to $i^*$ of length $n-1$. 
Let $\Pa_{j^*\to i^*}=((j^*,\ell_{1}),\dots, (\ell_{n-2},i^*))$. 
Consider $(i^*,j^*)$ was selected at any time $t$. 
Now, pick positive numbers $\epsilon_1,\dots,\epsilon_{n-1}$ such that $\sum^{n-1}_{k=1}\epsilon_k<\frac{U(t)}{2}$. Then, from the sentence below~\eqref{eq:geod-1}, we can first select $T_1$ times the edge $(\ell_{n-2},i^*)$ such that $W_2(\mu_{\ell_{n-2}}(t+T_1),\mu_{i^*}(t))<\epsilon_L$; then, we can select $T_2$ times the edge $(\ell_{n-3},\ell_{n-2})$ such that $W_2(\mu_{\ell_{n-3}}(t+T_1+T_2),\mu_{\ell_{n-2}}(t+T_1))<\epsilon_{n-2}$; and we can continue like this until finally selecting $T_{n-1}$ times the edge $(j^*,\ell_{1})$ such that $W_2(\mu_{j^*}(t+T),\mu_{\ell_{1}}(t+\sum^{n-2}_{k=1} T_k))<\epsilon_{1}$, with $T=\sum^{n-1}_{k=1} T_k$. Then, 
\begin{align*}
\sum_{(i,j)\in \Pa_{j^*\to i^*}} W_2(\mu_i(t+T),\mu_j(t+T))&= W_2\Big(\mu_{j^*}(t+T),\mu_{\ell_1}\big(\sum^{n-2}_{k=1} T_k\big)\Big)\\
&\quad+\sum_{m=1}^{n-3}W_2\Big(\mu_{\ell_m}\big(t+\sum_{k=1}^{n-1-m}T_k\big),\mu_{\ell_{m+1}}\big(t+\sum_{k=1}^{n-1-(m+1)}T_k\big)\Big)\\
&\quad+W_2(\mu_{\ell_{n-2}}(t+T_1),\mu_{i^*}(t))\\
&<\sum_{i=1}^{n-1} \epsilon_i< \frac{U(t)}{2}.
\end{align*}
Moreover, this result and the triangle inequality imply

which along the triangle inequality implies 
$$W_2(\mu_{i^*}(t+T),\mu_{j^*}(t+T))\leq\sum_{(i,j)\in \Pa_{j^*\to i^*}}W_2(\mu_i(t+T),\mu_j(t+T))<\frac{U(t)}{2},$$
and 
thus $U(t+T)=W_2(\mu_{i^*}(t+T),\mu_{j^*}(t+T))+\sum_{(i,j)\in \Pa_{j^*\to i^*}}W_2(\mu_{i}(t+T),\mu_j(t+T))<\frac{U(t)}{2}+\frac{U(t)}{2}=U(t)$. 
This implies the event ``$U(t+T)<U(t)$ for some finite $T>0$" has positive probability of happening at any time $t$ (because the finite sequence of edges described above has a positive probability of being selected sequentially at any time $t$), and so it can happen infinitely often with probability one. 
%
Therefore, we conclude that 
$U(t)\to 0$ as $t\to\infty$ with probability one. 
Then, $G$ being a cycle implies $U(t)=0$ iff $\mu_i(t)=\mu_j(t)$ for any $i,j\in V$, and the consensus result~\eqref{eq:conv-gen} follows. The particular value of the consensus measure $\mu_\infty$ is random since it 
may depend on the specific realization of the edge selection process. This finishes the proof for case~\ref{it:gen-1}.  

Finally, for the symmetric PaWBar algorithm in case~\ref{it:gen-2}, let $E=\{(1,2),\dots,(n-1,n)\}$  without loss of generality and set 
\begin{equation}
\label{eq:sep-s}
U(t)=\sum^{n-1}_{i=1}W_2(\mu_i(t),\mu_{i+1}(t)).
\end{equation}
Consider any $\{i,i+1\}\in E$ is selected at time $t$. In the following, consider this notation: for any $a,b\in\{t,t+1\}$ and $k\geq 1$, set $W_2(\mu_{1-k}(a),\mu_{1}(b))=0$ and $W_2(\mu_{n}(a),\mu_{n+k}(b))=0$. 
Then, setting $\mU(t)=\sum_{\substack{j=1\\j\neq i-1,i,i+1}}^n W_2(\mu_j(t),\mu_{j+1}(t))$, 
\begin{align*}
U(t+1)&=W_2(\mu_{i-1}(t),\mu_{i}(t+1))
+W_2(\mu_{i+1}(t+1),\mu_{i+2}(t))+\mU(t)\\
&\leq W_2(\mu_{i-1}(t),\mu_{i}(t))+W_2(\mu_{i}(t),\mu_{i}(t+1))+W_2(\mu_{i+1}(t+1),\mu_{i+1}(t))\\
&\quad +W_2(\mu_{i+1}(t),\mu_{i+2}(t))+\mU(t)\\
&=\frac{1}{2}W_2(\mu_{i}(t),\mu_{i+1}(t))+\frac{1}{2}W_2(\mu_{i}(t),\mu_{i+1}(t))+ \sum_{\substack{j=1\\j\neq i}}^n W_2(\mu_j(t),\mu_{j+1}(t))\\
&=U(t),
\end{align*}
where we used the triangle inequality and equation~\eqref{eq:geod-1}. 
Then $U(t+1)\leq U(t)$ with probability one. Following a similar analysis to case~\ref{it:gen-1}, assume the nontrivial case $U(t)\neq 0$. If $W_2(\mu_1(t),\mu_2(t))\neq 0$ or $W_2(\mu_{n-1}(t),\mu_n(t))\neq 0$, then it follows from our previous derivation that choosing the edge $\{1,2\}$ or $\{n-1,n\}$ at time $t$ implies $U(t+1)<U(t)$. Now if $W_2(\mu_1(t),\mu_2(t))=W_2(\mu_{n-1}(t),\mu_n(t))=0$ (obviously we consider $n\geq 4$ since for $n=2,3$ there is nothing to prove), 
then, it is easy to prove that we can select a finite sequence of edges, say of some length $T'$, such that $W_2(\mu_1(t+T'),\mu_2(t+T'))\neq 0$ or $W_2(\mu_{n-1}(t+T'),\mu_n(t+T'))\neq 0$. After such sequence is selected, we can select $\{1,2\}$ or $\{n-1,n\}$ so that $U(t+T'+1)<U(t+T')\leq U(t)$. Therefore, at any time $t$, the event ``$U(t+T)<U(t)$ for some finite $T>0$" has positive probability. Finally, following a similar analysis to case~\ref{it:gen-1}, we conclude that $U(t)\to 0$ as $t\to\infty$ with probability one and conclude the convergence proof of case~\ref{it:gen-2}. 
\end{proof}
\section{The relevance of the PaWBar algorithm in opinion dynamics}
\label{sec:app-ex}
In this section we discuss how the directed PaWBar algorithm generalizes a
well-known opinion dynamics model with real-valued beliefs to a model with
probability distributions as beliefs.  Assume the strongly-connected
weighted digraph $G=(V,E,A)$ describes a social network, whereby each agent
is an individual and the weight $a_{ij}\in(0,1)$, for each $(i,j)\in E$,
indicates how much influence individual $i$ accords to individual $j$.
Traditionally in the field of opinion dynamics, the opinion or belief of
any $i\in V$ at time $t$ is modeled as a scalar $x_i(t)\in \R$.  In the
popular \emph{asynchronous averaging model} (e.g., see~\cite{GD-DN-FA-GW:00,DA-AO:11})
beliefs evolve as follows: if $(i,j)\in E$ is selected at time $t$, then
$x_i(t+1)=(1-a_{ij})x_i(t)+a_{ij}x_j(t)$. Note that the PaWBar algorithm
specializes to the asynchronous averaging model (as a consequence of
Theorem~\ref{th:discr-gen}) when each agent
has a degenerate initial distribution with unit mass at a single scalar
value.

It is easy to formulate a second generalization of the asynchronous
averaging model.  Let $\mu_i(t)$ and $\mu_j(t)$ denote the beliefs of
individuals~$i$ and~$j$, assume $(i,j)\in E$ is selected at time $t$, and
consider the update $\mu_i(t+1)=(1-a_{ij})\mu_i(t)+a_{ij}\mu_j(t)$. This
second model is a simple (weighted) averaging of the beliefs; we call it
the \emph{AoB model}. To understand the similarities and difference between
the PaWBar and AoB models, assume the beliefs of individuals $i$ and $j$ at
time $t$ are Gaussian distributions $\mathcal{N}(x_i(t),\sigma)$ and
$\mathcal{N}(x_j(t),\sigma)$ with equal variance.  Under this assumption,
one can see that both models predict that $i$'s mean opinion evolves
according to $x_i(t+1)=(1-a_{ij})x_i(t)+a_{ij}x_j(t)$. However, the two
models differ in the predicted overall belief and, specifically:
\begin{alignat}{3} 
  & \text{PaWBar model:} \qquad & \mu_i(t+1)&:= \mathcal{N}\big((1-a_{ij})x_i(t)+a_{ij}x_j(t),\sigma\big), \\
  & \text{AoB model:} \qquad & \mu_i(t+1)&:=
  (1-a_{ij})\mathcal{N}(x_i(t),\sigma)+a_{ij}\mathcal{N}(x_j(t),\sigma).
\end{alignat}
In other words, the PaWBar model predicts a Gaussian belief and the AoB
model predicts a Gaussian mixture belief. Even though both resulting
beliefs have the same mean, they overall differ substantially.

Finally, we argue that the PaWBar algorithm is preferable over the AoB
model for opinion evolution from a cognitive psychology viewpoint.  In the
case of initial Gaussian beliefs, the PaWBar algorithm dictates that $i$'s
belief is simply Gaussian at every time. Thus, as $i$ continues her
interactions in the social network, the memory cost associated to her
belief at all times is constant: $i$ remembers only two scalars, i.e., the
mean opinion and its variance. Instead, if $i$ updates her belief according
to the AoB model, then her belief is a Gaussian mixture at every time and
$i$ is required to remember a more complicated belief structure.  Thus, the
AoB model implies that $i$ requires more cognitive power and memory to
process the information she gathers from her interactions.  The problem
with the AoB approach is that arguably individuals tend to simplify beliefs
in order to both remember and process thoughts more economically. This
simplification of beliefs has attributed humans the metaphor of being
\emph{cognitive misers} in cognitive
psychology~\cite{SF-SET:17,PJO-JCT:90}.  Therefore, a model with more
economic belief memory requirements, such as our PaWBar algorithm, is
arguably more adequate.

\section{Conclusion}
\label{sec:concl}

We propose the PaWBar algorithm based on stochastic asynchronous pairwise
interactions. For specific classes of discrete and absolutely continuous
measures, we characterize the computation of both randomized and standard
Wasserstein barycenters under arbitrary graphs. For the case of general
measures, we prove a consensus result and leave the existence of a barycenter
as an open problem. We also
specialize our algorithm to the Gaussian case and establish a relationship
with models of opinion dynamics.

We hope our paper elicits research on efficient numerical solvers for the
distributed computation of Wasserstein barycenters based on pairwise
computations. 
%
As future work, given the importance of Gaussian distributions, we
envision theoretical progress in proving the conjecture proposed in our
paper. Another open problem is to design consensus algorithms that
  guarantee the \emph{exact} computation of a desired weighted Wasserstein
  barycenter through asynchronous pairwise computations, an unsolved
  problem presented in~\cite{ANB-AD:14}.

\bibliographystyle{plainurl}
\bibliography{alias,Main,FB}

\end{document}